%% file: arxiv.tex
\relax
\documentclass{article}
\usepackage{ijcai19}  
\usepackage{times}
\usepackage{soul}
\usepackage{url}
\usepackage[hidelinks]{hyperref}
\usepackage[utf8]{inputenc}
\usepackage[small]{caption}
\usepackage{graphicx}
\usepackage{amsmath}
\usepackage{booktabs}
\usepackage{wrapfig}
\usepackage{misc}
\usepackage{xspace}
\usepackage{stmaryrd}
\usepackage{enumitem}
\usepackage{bbm}
\usepackage{amssymb,amsthm,amsmath}
\usepackage{color}
\usepackage{centernot} 
\usepackage[boxed]{algorithm2e}
\DeclareMathSymbol{\shortminus}{\mathbin}{AMSa}{"39}
\usepackage{tikz}
\usetikzlibrary{fit,calc}
\usetikzlibrary{arrows,decorations.pathmorphing,backgrounds,positioning,fit,petri}

\newcommand{\mn}[1]{\ensuremath{\mathsf{#1}}}
\newcommand{\qnrleq}[3]{\ensuremath{(\leqslant #1 \; #2 \; #3)}}
\newcommand{\qnrgeq}[3]{\ensuremath{(\geqslant #1 \; #2 \; #3)}}
\newcommand{\qnreq}[3]{\ensuremath{(= #1 \; #2 \; #3)}}

\newcommand{\Aa}{{\cal{A}}}

\newcommand{\Ii}{{\cal{I}}}
\newcommand{\Jj}{{\cal{J}}}
\newcommand{\Kk}{{\cal{K}}}

\newcommand{\Mm}{{\cal{M}}}

\newcommand{\Tt}{{\cal{T}}}

\newcommand{\notmodels}{\centernot\models}

\newcommand{\ALCO}{\ensuremath{{\cal{ALC\hspace{-0.25ex}O}}}\xspace}

\newcommand{\SOF}{\ensuremath{{\cal{S\hspace{-0.2ex}O\hspace{-0.24ex}F}}}\xspace}

\newcommand{\SQ}{\ensuremath{{\cal{S\hspace{-0.06ex}Q}}}\xspace}
\newcommand{\SOQ}{\ensuremath{{\cal{S\hspace{-0.06ex}O\hspace{-0.1ex}Q}}}\xspace}
\newcommand{\SIQfwd}{\ensuremath{{\cal{S\hspace{-0.06ex}I\hspace{-0.1ex}Q}^\shortminus}}\xspace}

\newcommand{\bagz}{\ensuremath{\Imf}\xspace}
\newcommand{\bago}{\ensuremath{\mathfrak{r}}\xspace}

\newcommand{\concepts}{\ensuremath{\mn{N_\mn{C}}}}
\newcommand{\roles}{\ensuremath{\mn{N_{R}}}}
\newcommand{\troles}{\ensuremath{\mn{N_{R}^{t}}}}
\newcommand{\ntroles}{\ensuremath{\mn{N_{R}^{nt}}}}
\newcommand{\rel}{\mn{rel}}

\newcommand{\Nomi}{\mn{nom}}
\newcommand{\Ind}{\mn{ind}}

\newcommand{\twoexp}{\ensuremath{\textsc{2ExpTime}}\xspace}

\usepackage{thmtools,thm-restate}

\newtheorem{definition}{Definition}
\newtheorem{theorem}{Theorem}

\newtheorem{lemma}{Lemma}

\newtheorem{claim}{Claim}
\allowdisplaybreaks

\title{On Finite and Unrestricted Query Entailment beyond \SQ with Number Restrictions on Transitive Roles}

\author{
Tomasz Gogacz$^1$\and
 V\'ictor Guti\'errez-Basulto$^2$\and
Yazm\'in Ib\'a\~nez-Garc\'ia$^{2}$\and \\
Jean Christoph Jung$^3$ \And 
Filip Murlak$^1$
\\
\affiliations
$^1$University of Warsaw, Poland\\
$^2$Cardiff University, UK\\
$^3$University of Bremen, Germany\\
\emails     
t.gogacz@mimuw.edu.pl, 
\{gutierrezbasultov, ibanezgarciay\}@cardiff.ac.uk,\\
jeanjung@uni-bremen.de, 
fmurlak@mimuw.edu.pl
}

\begin{document}

\maketitle 
\begin{abstract}
%

 We study the description logic \SQ with number restrictions
applicable to transitive roles, extended with either \emph{nominals}
or \emph{inverse roles}.  We show tight \textsc{2ExpTime} upper bounds
for \emph{unrestricted} entailment of \emph{regular path queries} for
both extensions and \emph{finite} entailment of \emph{positive
existential queries} for nominals.  For inverses, we establish
\twoexp-completeness for \emph{unrestricted} and \emph{finite}
entailment of \emph{instance queries} (the latter under restriction to
a single, transitive role).

\end{abstract}

\input{introduction}

\input{preliminaries}

\input{tree-likeMP}

\input{omqa}

\input{fomqa-soq}

\input{finsat-siq}

\input{conclusions}

\bibliographystyle{named}

\bibliography{references}

\clearpage
\appendix

\input{additional-notation}
\input{tree-likeMP-appendix}

\input{omqa-appendix}
\input{fomqa-soq-appendix}

\input{finsat-siq-appendix}

\end{document}

%% file: introduction.tex
\section{Introduction}

A prominent line of research in knowledge representation and database
theory has focused on the evaluation of queries over incomplete data
enriched by ontologies
providing background knowledge.
In this paradigm, ontologies are commonly formulated using
description logics (DLs),
believed to offer a good balance between expressivity and complexity.
This
is supported, for instance, by
the good understanding of `data-tractable' DLs~\cite{KontchakovZ14,BienvenuO15}. 
Yet, for some expressive DLs  the complexity of query entailment is
less understood.


In this paper, we study query entailment in  extensions of the description
logic (DL) \SQ
allowing number restrictions (\Qmc) to be applied to transitive roles
(\Smc). Most previous work on query entailment in expressive DLs, such
as $\mathcal{SHIQ}$ or $\mathcal{SHOQ}$, forbid the interaction of
number restrictions and transitive
roles~\cite{GlimmLHS08,GlimmHS08,CalvaneseEO14}, but it is required in
areas like biomedicine, e.g., to restrict the number of certain parts
an organ has. For instance, one can express that the human heart has
exactly one mitral valve, which has to be shared by its left and right
atrium~\cite{GuIbJu-AAAI18}. Allowing for the interaction of \Smc and
\Qmc is dangerous in the sense that even modest extensions of \SQ,
such as with role inclusions or inverse roles, lead to an undecidable
satisfiability problem~\cite{DBLP:conf/lpar/KazakovSZ07}.
%
%
%
%
Decidability of satisfiability in \SQ and in its extension with
nominals was shown several years
ago~\cite{DBLP:conf/lpar/KazakovSZ07,KaminskiS10}, but only recently
tight computational complexity bounds were
established~\cite{GuIbJu-AAAI17}. 
Even more recently, decidability for entailment of regular path
queries over \SQ knowledge bases was established. More precisely,
based on a novel \emph{tree-like model property} of \SQ it was
possible to devise an automata-based decision procedure yielding a
tight \twoexp upper bound~\cite{GuIbJu-AAAI18}.

The objective of this paper is to provide a more complete picture of
query entailment in DLs with number restrictions on transitive roles.
We pursue two specific goals.

First, we aim at understanding the
limits of decidability of query entailment for such DLs.  To
this end, we investigate the extensions of \SQ by \emph{nominals
(\SOQ)} and \emph{controlled inverse roles (\SIQfwd)}, where we
allow number restrictions on inverse non-transitive roles and only
existential restrictions on inverse transitive roles. 
%
As query language, we consider \emph{positive existential regular path
queries}, thus capturing the common languages of conjunctive 
and regular path queries.

Our second aim is to initiate the study of \emph{finite} query
entailment for \SIQfwd and \SOQ, where one is interested in reasoning
only over finite models.  This distinction is crucial because in
database applications, both database instances and the models they
represent are commonly assumed to be finite. The study of finite query
entailment in \SQ is interesting since, due to the presence of
transitivity, \SQ lacks \emph{finite controllability}, and therefore
unrestricted and finite entailment do not coincide. Interestingly,
most previous works on finite query entailment consider logics lacking
finite controllability because of number restrictions and inverse
roles~\cite{Rosati08,Pratt-Hartmann09,GarciaLS14,AmarilliB15}.  The
study of finite query entailment in logics with transitivity (without
number restrictions on transitive roles) started only recently
\cite{Rudolph16,GogaczIM18,Kieronski18}. Here, we focus on finite
entailment of positive existential queries in \SOQ and of instance
queries in \SIQfwd.
 
Our main contributions are as follows.  In
Sect.~\ref{sec:tree-likeMP}, we start by showing a \emph{tree-like
model property} for both \SOQ and \SIQfwd.  More specifically, we
carefully extend and adapt the \emph{canonical} tree decompositions
that were introduced for \SQ in previous work~\cite{GuIbJu-AAAI18} to
also incorporate the
presence of controlled inverses and nominals. Next, we prove that if a
query is not entailed by a knowledge base (KB), then there is a
counter-model with a canonical tree decomposition of small width.
%
%
%
%
This tree-like model property is the basis for automata-based
approaches to unrestricted and finite query entailment in the
remainder of the paper.
%
First, in Sect.~\ref{sec:omqa}, we construct tree automata to
optimally decide entailment of regular path queries over \SOQ and
\SIQfwd KBs in \twoexp. We move then, in Sect.~\ref{sec:fomqa-soq},
to finite entailment of positive existential queries over \SOQ KBs,
showing again an optimal \twoexp upper bound. To this end, we look at
more refined canonical tree decompositions, which ensure the
existence of a finite counter model.  In other words, we reduce finite
query entailment to entailment over models with this special canonical
tree decomposition. Finally, in Sect.~\ref{sec:finsat-siq}, we
investigate the complexity for unrestricted and finite instance query
(IQ) entailment in \SIQfwd. In particular, we show that IQ entailment
is \twoexp-hard both in the finite and in the unrestricted case. We
found this surprising since it is rarely the case that IQ entailment
becomes more difficult when inverses are added to the logic. Moreover,
the result provides an orthogonal reason for \twoexp-hardness for
conjunctive query entailment in \SIQfwd~\cite{DBLP:conf/cade/Lutz08}.
We complement this lower bound with matching upper bounds in the
unrestricted case, thus confirming the conjecture that satisfiability
in \SIQfwd is decidable~\cite{DBLP:conf/lpar/KazakovSZ07}. In the
finite case, we show a \twoexp-upper bound for KBs using a single
transitive role. Note that \SIQfwd with a single transitive role is a
notational variant of the \emph{graded modal logic with converse
$\mathbf{K4}(\Diamond_{\geq}, \Diamond^{\shortminus})$}. Thus, our
result entails \twoexp-completeness for global consequence in $\mathbf{K4}(\Diamond_{\geq},
\Diamond^{\shortminus})$, which was only known to be decidable~\cite{BKieronski18}.

A long version with appendix can be found under
\url{http://www.informatik.uni-bremen.de/tdki/research/papers.html}.
 

%% file: preliminaries.tex
\section{Preliminaries}
\label{sec:preliminaries}
 
\subsection{Description Logics} 

We consider a vocabulary consisting of
countably infinite disjoint sets of \emph{concept names} $\mn{N_C}$,
\emph{role names} $\mn{N_R}$, and \emph{individual names} $\mn{N_I}$,
and assume that $\mn{N_R}$ is partitioned into two infinite
sets of \emph{non-transitive role names} $\ntroles$ and
\emph{transitive role names} $\troles$. A \emph{role} is a role name 
or an \emph{inverse role} $r^-$; a \emph{transitive role} is a transitive
role name or the inverse of one.
%
\emph{\SIQfwd-concepts $C,D$} are defined by the grammar
 $$C,D ::= A\mid \neg C \mid C \sqcap D \mid \exists r. C  \mid  \qnrleq n s C  $$ 
where $A \in \mn{N_C}$,
$r$ is a role, $n \geq 0$ is a natural number given in binary, and $s$ is
either a non-transitive role or a transitive role name. 
 \emph{\SOQ-concepts $C,D$} are defined by the grammar
  $$C,D ::= A\mid \neg C \mid C \sqcap D \mid \{a\}  \mid   \qnrleq n r C$$ 
where $A \in \mn{N_C}$,
$r \in \mn{N_R}$, $a \in \mn{N_I}$ and $n$ is as above. 
We will use $\qnrgeq n r C$ as  abbreviation for $\neg \qnrleq {n{-}1}
r C$, together with standard abbreviations $\bot$, $\top$, $C\sqcup
D$, $\forall r.C$. 
Concepts of the form $\qnrleq n  r  C $, $\qnrgeq n r C$, and
$\{a\}$ are called 
\emph{at-most restrictions}, \emph{at-least restrictions}, and
\emph{nominals}, respectively. Note that in $\SIQfwd$ concepts, \emph{inverse
  transitive roles} are not allowed in at-most and at-least restrictions.

A \emph{$\SIQfwd$-TBox (respectively, \SOQ-TBox) \Tmc} is a finite set
of \emph{concept inclusions (CIs)} $C\sqsubseteq D$, where  $C,D$ are
\SIQfwd-concepts (respectively, \SOQ-concepts).  An \emph{ABox} $\Amc$
is a finite non-empty set of \emph{concept} and \emph{role assertions} of the
form $A(a)$, $r(a,b)$ where $A \in \mn{N_C}$, $r \in \mn{N_R}$ and
$\{a,b\} \subseteq \mn{N_I}$; $\mn{ind}(\Amc)$ is the set of 
individual names occurring in \Amc. A \emph{knowledge base (KB)} is a
pair $\Kk=(\Tt, \Aa)$; $\mn{nom}(\Kmc)$ is the set of 
nominals occurring in $\Kmc$ and $\mn{ind}(\Kmc) = \mn{ind}(\Amc) \cup
\mn{nom}(\Kmc)$.

Without loss of generality, we assume throughout the paper that all
CIs are in one of the following \emph{normal forms}:
\begin{align*}
  \textstyle\bigsqcap_i A_i \sqsubseteq \bigsqcup_j B_j,
  \quad A \sqsubseteq \forall r^-. B, \quad A\sqsubseteq \exists
  r^-.B, \\ 
  \quad A \sqsubseteq \qnrleq n s B,\quad A \sqsubseteq \qnrgeq n s B,
\end{align*}
where $A,A_i,B,B_j$ are concept names or nominals, $r\in\roles$, $s$
is a non-transitive role or a transitive role name, 
and empty disjunction and conjunction are
equivalent to $\bot$ and $\top$, respectively. We further assume
that for every at-most and at-least restriction, \Tmc contains an
equivalent concept name. 

\subsection{Interpretations} 

The semantics is given as usual via
\emph{interpretations} $\Imc = (\Delta^\Imc, \cdot^\Imc)$ consisting
of a non-empty \emph{domain} $\Delta^\Imc$ and an \emph{interpretation
function $\cdot^\Imc$} mapping concept names to subsets of the domain
and role names to binary relations over the domain. Further, we adopt
the \emph{standard name assumption}, i.e., $a^\Imc=a$ for all $a \in
\mn{N_I}$.  The interpretation of complex concepts $C$ is defined in
the usual way~\cite{DLBook}.  An interpretation $\Imc$ is a
\emph{model of a TBox \Tmc}, written $\Imc\models\Tmc$ if
$C^\Imc\subseteq D^\Imc$ for all CIs $C\sqsubseteq D\in \Tmc$.  It is
a \emph{model of an ABox \Amc}, written $\Imc\models \Amc$, if
$(a,b)\in r^\Imc$ for all $r(a,b)\in \Amc$ and $a\in A^\Imc$ for all
$A(a)\in \Amc$. Finally, $\Imc$ is a \emph{model of a KB
$\Kk=(\Tt, \Aa)$}, written $\Imc\models \Kmc$, if $\Ii\models \Tt$,
$\Imc\models \Aa$, and $r^\Ii$ is transitive for all $r\in\troles$ occurring in $\Kmc$.
If \Kmc has a model, we say that it is \emph{satisfiable}.

An interpretation $\Imc'$ is a
\emph{sub-interpretation} of $\Imc$, written as $\Imc'\subseteq \Imc$,
if $\Delta^{\Imc'}\subseteq \Delta^\Imc$, $A^{\Imc'}\subseteq A^\Imc$,
and $r^{\Imc'}\subseteq r^{\Imc}$ for all $A\in\mn{N_C}$ and
$r\in \mn{N_R}$. For $\Sigma \subseteq \concepts \cup \roles$,  $\Ii$
is a \emph{$\Sigma$-interpretation} if $A^\Ii=\emptyset$ and
$r^\Ii=\emptyset$ for all $A\in \concepts \setminus \Sigma$ and $r\in
\roles\setminus\Sigma$. The \emph{restriction of $\Imc$ to signature
$\Sigma$} is the maximal $\Sigma$-interpretation $\Imc'$ with
$\Imc'\subseteq \Imc$.
%
The \emph{restriction of $\Imc$ to
domain $\Delta$} is the maximal sub-interpretation of $\Imc$ with
domain $\Delta$. The union $\Ii \cup \Jj$ of $\Ii$ and $\Jj$ is an
interpretation such that $\Delta^{\Ii \cup \Jj} = \Delta^{\Ii} \cup
\Delta^{\Jj}$, $A^{\Ii \cup \Jj} = A^{\Ii} \cup
A^{\Jj}$, and $r^{\Ii \cup \Jj} = r^{\Ii} \cup
r^{\Jj}$ for all $A\in\mn{N_C}$ and $r\in \mn{N_R}$.
The \emph{transitive closure} $\Ii^*$ of $\Ii$ is an interpretation such that
$\Delta^{\Ii^*} = \Delta^{\Ii}$, $A^{\Ii^*} = A^{\Ii}$ for all
$A\in\mn{N_C}$, $r^{\Ii^*} = r^{\Ii}$ for all $r\in \ntroles$, and
$r^{\Ii^*} = (r^{\Ii})^+$ for all $r\in \troles$.

%
A \emph{tree decomposition $\Tmf$ of an interpretation~\Imc} is a pair
$(T,\Imf)$ where $T$ is a tree and $\Imf$ is a function that assigns
an interpretation $\Imf(w) = (\Delta_w, \cdot^{\Imf(w)})$ to each
$w\in T$ such that $\Ii = \bigcup_{w \in T}
\Imf(w)$ and for every $d\in \Delta^\Imc$, the set $\{w\in T\mid d\in
\Delta_w\}$ is connected in $T$.
We often blur the distinction between a node $w$ of $T$ and
the associated interpretation $\Imf(w)$, using the term \emph{bag} for both.
The \emph{width} of $\Tmf$ is
$\sup_{w\in T}|\Delta_w| - 1$; the \emph{outdegree} of $\Tmf$ is the
outdegree of $T$.
For each $d\in\Delta^\Imc$, there is a unique bag $w$ closest to the
\emph{root} $\varepsilon$ such
that $d\in \Delta_w$. We say that $d$ is \emph{fresh} in this bag,
and write $F(w)$ for the set of all elements fresh in $w$.

\subsection{Ontology-mediated Query Entailment}

A \emph{positive existential regular path query (PRPQ)} is a
first-order formula $\varphi=\exists\xbf\, \psi(\xbf)$ with
$\psi(\xbf)$ constructed using $\wedge$ and $\vee$ over
atoms of the form $\Emc(t,t')$ where $t,t'$ are variables from
$\xbf$ or individual names from $\mn{N_I}$, and \Emc is a \emph{path
expression} defined by the grammar
$$\Emc,\Emc' ::= r \mid r^- \mid A?\mid \Emc^* \mid \Emc \cup \Emc'
\mid \Emc\circ\Emc',$$
where $r\in \mn{N_R}$ and $A\in \mn{N_C}$. 
A PEQ is a PRPQ that does not use the operators $^*$, $\cup$, and
$\circ$ in path expressions. Equivalently, it is an FO formula
$\varphi=\exists \xbf\,\psi(\xbf)$ where $\psi$ is constructed using
$\wedge$ and $\vee$ over atoms $r(t,t')$ and $A(t')$ with
$t,t'$ as above. An \emph{instance query (IQ)} is just an
expression of the shape $C(a)$ for some concept $C$ and
$a\in\mn{N_I}$.

The semantics of PRPQs is defined via matches. Let us fix a PRPQ
$\varphi=\exists \xbf\,\psi(\xbf)$ and an interpretation \Imc.  Let
$\mn{ind}(\varphi)$ be the set of individual names in
$\varphi$. A \emph{match for $\vp$ in \Imc} is a function
$\pi:\xbf\cup \mn{ind}(\varphi)\to \Delta^\Imc$ such that $\pi(a)=a$,
for all $a\in \mn{ind}(\varphi)$, and $\Imc,\pi\models\psi(\xbf)$
under the standard semantics of first-order logic extended with a rule
for atoms of the form $\Emc(t,t')$.  An interpretation $\Imc$
satisfies $\vp$, written as $\Imc\models\vp$, if there is a match for
$\vp$ in \Imc.

A PRPQ $\vp$ is \emph{(finitely) entailed by a
KB \Kmc}, if $\Imc\models\vp$ for every (finite) model
\Imc of \Kmc; we write $\Kmc\models\vp$ and $\Kmc\models_{\mn{fin}}\vp$,
respectively, in this case. Accordingly, we write $\Kmc\models
C(a)$ and $\Kmc\models_{\mn{fin}} C(a)$ if 
$a\in C^\Imc$ in all (finite) models
$\Imc$ of \Kmc.

We study the corresponding \emph{decision problem}---whether a given query is
(finitely) entailed by a given KB---for different choices
of knowledge base and query languages.



%% file: tree-likeMP.tex
\section{Tree-like Counter-Model Property}
\label{sec:tree-likeMP}

In this section we show a tree-like model property for $\SIQfwd$ and
$\SOQ$: we show that if a query is not entailed by a KB, then there is
a counter-model with a tree decomposition of bounded width and
outdegree. For the automata-based decision procedure to yield
optimal upper bounds, it is useful to consider \emph{canonical}
decompositions which we define next.

In canonical decompositions elements will be accompanied by certain key
neighbors.
Let us fix a KB $\Kk=(\Tt, \Aa)$.
For an interpretation $\Imc$, an  element $d\in \Delta^\Imc$, and $r \in
\mn{N}_\mn{R}^{t}$, the \emph{$r$-cluster of $d$ in \Imc}, denoted by
$Q^\Imc_r(d)$, is the set containing $d$ and each $e\in\Delta^\Imc$
such that both $(d,e)\in r^\Imc$ and $(e,d)\in r^\Imc$. This is the
closest environment of $d$ wrt.~$r$.
%
%
We also associate with $d$ a larger
set $\rel^\Ii_r(d)$ of $r$-successors \emph{relevant for the at-most  
restrictions of $\Kk$}. We let $\rel^\Ii_r(d)$ be the least set $X$ such that
$Q^\Ii_r(d) \subseteq X$ and for all $e \in X$, $f\in\Delta^\Ii$, and
$A \sqsubseteq \qnrleq n r B$ in  $\Tmc$, if $e\in A^\Imc$, 
$f \in B^\Imc$, and $(e,f)\in
r^{\Imc^*}$, then $Q^\Ii_r(f) \subseteq X$.
%
%
%
The construction of canonical decompositions relies on the
  following properties of relevant successors.

\begin{restatable}{lemma}{lemrelevant}
\label{lem:relevant}
  For each $r\in\troles$, the following hold:
  \begin{enumerate}
  \item \label{it:relevant1}
    for all $d,e\in\Delta^\Ii$, if $e \in \rel^\Ii_r(d)$ then
    $\rel^\Ii_r(e) \subseteq \rel^\Ii_r(d)$; 
  \item \label{it:relevant2}
    if each $r$-cluster in $\Ii$ has size at most $N$, then for
    each $d\in \Delta^\Imc$, $|\rel^\Ii_r(d)|\leq N \cdot 2^{\mn{poly}(|\Tmc|)}$. 
  \end{enumerate}
\end{restatable}

In a canonical tree decomposition, formalized in
Definition~\ref{def:canonical} below, each non-root bag keeps track of
all concepts and a single role indicated by
$\bago$. Nominals are captured within a finite subinterpretation $\Mm$
represented faithfully in all bags; in the absence of nominals, one
can take empty \Mmc and drop~\ref{it:can4}. Conditions
\ref{it:can0}--\ref{it:can3} ensure that apart from $\Delta^\Mm$,
neighboring non-root bags share a single element, sometimes
accompanied by its relevant successors.


\begin{definition}\label{def:canonical}
  A tree decomposition $\Tmf = (T,\Imf)$ is \emph{canonical} if
there exists $\bago \colon T \to \roles \cup \{\bot\}$ with
  $\bago^{-1}(\bot) = \{\varepsilon\}$ such that
\begin{enumerate}[label=(B$_\arabic*$),leftmargin=*,align=left]
\setcounter{enumi}{-1}
\item\label{it:canbase0}
  for each $w\in T$, $\Imf(w)$ is a $\Sigma_w$-interpretation where
  $\Sigma_\varepsilon=\concepts\cup\ntroles$ and $\Sigma_w = \concepts
  \cup \{\bago(w)\}$
  for $w\neq\varepsilon$;
\item\label{it:canbase1}
  for all $v,w \in T$, the restrictions of $\Imf(v)$ and $\Imf(w)$ to
  domain $\Delta_v\cap \Delta_w$ and signature $\Sigma_v\cap\Sigma_w$ coincide;
\item\label{it:canbase3}
  for each $v \in T \setminus \{\varepsilon\}$, $d\in F(v)$, and
  $r\in\troles \setminus \{\bago(v)\}$, a unique child $w$ of $v$ satisfies
  $\bago(w)=r$ and $d \in  \Delta_{w}$;
\end{enumerate}  
and there is an interpretation $\Mm$ with 
$\Nomi(\Kk) \subseteq \Delta^{\Mm}\subseteq\Delta_\varepsilon$, such that for each $w\in
T\setminus\{\varepsilon\}$ and its parent $v$, one has
\begin{enumerate}[label=(C$_\arabic*$),leftmargin=*,align=left]
\setcounter{enumi}{-1}

\item\label{it:can0}
  if $\bago(w)\in \troles$ and $v = \varepsilon$, then 
  $ \Delta_\varepsilon\subseteq\Delta_w$;  

\item\label{it:can1}
  if $\bago(w) \in \ntroles$, then $\Delta_v\cap \Delta_w = \{d\} \cup
  \Delta^{\Mm}$ for some  $d\in F(v)$;

\item\label{it:can2}
 if $\bago(v)\neq \bago(w)\in\troles$ and $v\neq \varepsilon$, then
  $\Delta_w\cap\Delta_{v}=\{d\} \cup \Delta^{\Mm}$ for some $d\in F(v)$;

\item \label{it:can3} if $\bago(w) = \bago(v)= r \in \troles$, then
  $\Delta_v \cap \Delta_w = \rel^{\Imf(v)}_{r}(d) \cup \Delta^{\Mm}$
  and $\rel^{\Imf(v)}_r(d)=\rel^{\Imf(w)}_r(d)$ for some $d$ such that
  either $d\in F(v)$ or $d \in F(u)$ and  $\bago(u)\neq\bago(v)$ for
  the parent $u$ of $v$; and 

\item \label{it:can4} if $\bago(w)= r \in
  \troles$, then $\rel^{\Imf(w)}_{r}(d)=\rel^{\Mm}_r(d)$ for all $d
  \in \Delta^{\Mm}$.  
  
\end{enumerate} 
\end{definition}

\begin{restatable}{theorem}{thmtreelike} 
  \label{thm:tree-like}
  Let $\Kmc=(\Tmc, \Amc)$ be a KB in normal form and $\varphi$ a
  PRPQ with $\Kmc\notmodels\varphi$. If $\Kk$ is a $\SOQ$ KB or a $\SIQfwd$
  KB, then  there exists a model $\Jmc$ of $\Tmc$ and $\Amc$ such that  
  \begin{itemize}
  \item $\Jj$ has a canonical  tree decomposition of width and
    outdegree $\mn{poly}(|\Ind(\Kk)|)\cdot  2^{\mn{poly}(|\Tmc|)}$; and
  \item $\Jj^* \models \Kk$ and $\Jj^* \notmodels \varphi$.\footnote{Recall that in
      a model of the ABox or the TBox, the extensions of
      role names from $\troles$ need not be transitive.} 
  \end{itemize}
\end{restatable}

\begin{proof}
Let us fix a counter-model $\Imc$ for $\Kmc$ and $\varphi$.
We can assume that
$|Q^{\Ii}_r(d)| \leq |\Ind(\Kk)|+2^{\mn{poly}({|\Tmc|})}$ for all
$d\in\Delta^{\Ii}$~\cite{GuIbJu-AAAI18}.
By Lemma~\ref{lem:relevant}, $|\rel^\Ii_r(d)|
\leq |\Ind(\Kk)|\cdot 2^{\mn{poly}(|\Tmc|)}$ for all $d\in\Delta^\Ii$.

To build a canonical tree decomposition $\Tmf$, we unravel
$\Ii$ starting from the interpretation of the ABox and then applying
the extension rules $(\mathbf{R}_0)$--$(\mathbf{R}_3)$ below, 
corresponding to conditions \ref{it:can0}--\ref{it:can3}:
$(\mathbf{R}_0)$ collects relevant successors of the individuals in
the ABox, $(\mathbf{R}_1)$ performs standard unraveling of
non-transitive roles, $(\mathbf{R}_2)$ takes care of the change of roles,
and $(\mathbf{R}_3)$ realizes further unraveling of transitive roles.

More precisely, for the root bag, we take $\Ii$ restricted to the domain 
$\Ind(\Aa) \cup \Delta$ and the signature $\concepts \cup \ntroles$,
where $\Delta$ is the union of $\rel_r^\Ii(a)$ for all $a\in
\Nomi(\Kk)$ and $r\in \troles$. 

$(\mathbf{R}_0)$ For each $r\in\troles$, we add as a child bag of $\varepsilon$
the restriction of $\Ii$ 
to signature $\concepts \cup \{r\}$ and domain $\bigcup_{a \in
\Ind(\Aa)}\rel^\Ii_r(a) \cup \Delta$ with each $e\notin
\Ind(\Aa) \cup \Delta$ replaced by a fresh copy $e'$. We call $e$
the \emph{original of $e'$}.

Then, we use the following rules $(\mathbf{R}_1)$--$(\mathbf{R}_3)$
\emph{ad infinitum}, applying each rule only once to each previously
added bag $v$. 

$(\mathbf{R}_1)$ For each $r\in\ntroles$, and each $d' \in F(v)$, %
fresh in bag $v$, let $d\in\Delta^\Ii$ be the original of $d'$
(possibly $d=d'$) and let $W_0$ be the set of originals of all
$r$-successors and $r$-predecessors of $d'$ in $\Imf(v)$.  Pick a
minimal set $W \subseteq \Delta^\Ii$ containing $\{d\} \cup W_0 \cup
\Delta$ such that for each $s\in \{r, r^-\}$ and $A \sqsubseteq
\qnrgeq n s B$ in $\Tt$, if $d \in A^{\Ii}$, then $d$ has at least $n$
different $s$-successors in $B^\Ii \cap W$.  For each $e \in W
\setminus (W_0 \setminus \Delta)$, add as a child bag of $v$ the
restriction of $\Ii$ to signature $\concepts \cup \{r\}$ and domain
$\{d,e\} \cup \Delta$ with all $r$-edges from $\Delta \setminus\{d\}$
to $\{d,e\}\setminus \Delta$ removed, $d$ replaced by $d'$ and each $
f\in \{e\} \setminus \Delta$, by a fresh copy $f'$.

$(\mathbf{R}_2)$ Assuming $v\neq \varepsilon$, for each $r\in\troles$
with $r \neq \bago(v)$, and each $d' \in F(v)$, let $d$ be the
original of $d'$.  Add as a child bag of $v$ the restriction of $\Ii$
to signature $\concepts\cup\{r\}$ and domain $\rel^\Ii_r(d) \cup
\Delta$ where $d$ is replaced by $d'$ and each $e\in \rel^\Ii_r(d)
\setminus(\{d\}\cup \Delta)$, by a fresh copy $e'$.

$(\mathbf{R}_3)$ Assuming $\bago(v)=r \in \troles$, for each $d' \in
\Delta^v$ fresh in $v$ or in the parent $u$ of $v$ with $\bago(u)\neq
r$, let $d$ be the original of $d'$.  Pick a minimal set $W \subseteq
\Delta^\Ii$ containing $\rel^\Ii_r(d) \cup \Delta$ such that for each
$A \sqsubseteq \qnrgeq n r B$ in $\Tt$, if $d \in A^{\Ii}$, then $d$
has at least $n$ different $r$-successors in $B^\Ii \cap W$, and for
each $A \sqsubseteq \exists r^-. B$ in $\Tt$, if $d \in A^{\Ii}$, then
$d$ has an $r^-$-successor in $B^\Ii \cap W$.  For each $e\in W
\setminus (\rel^\Ii_r(d) \cup \Delta)$, add as a child bag of $v$ the
restriction of $\Ii$ to the signature $\concepts\cup\{r\}$ and domain
$\rel^\Ii_r(e) \cup \rel^\Ii_r(d) \cup \Delta$ where each element
$f\in \rel^\Ii_r(d) \setminus \Delta$ is replaced by its copy $f'$
from $\Imf(v)$, and each element $f\in \rel^\Ii_r(e) \setminus
(\rel^\Ii_r(d) \cup \Delta)$ by a fresh copy $f'$.

\smallskip

Let $\Jj$ be the interpretation underlying the resulting decomposition
$\Tmf$. The function mapping each $d' \in \Delta^\Jj$ to its original
$d\in\Delta^\Ii$ gives a homomorphism from $\Jj$ to $\Ii$, and
consequently also from $\Jj^*$ to $\Ii$. It follows that $\Jj^*
\notmodels \varphi$. Taking $\Ii$ restricted to $\Delta$ as $\Mm$, it
is routine to check that $\Tmf$ and $\Jj$ satisfy the remaining
postulated properties.
%
%
%
Note that while the construction
is described for any normalized $\Kk$, $(\mathbf{R}_1)$ is correct
only if $\Kk$ is either a $\SOQ$ KB or a $\SIQfwd$ KB. Correctness of
$(\mathbf{R}_2)$ and $(\mathbf{R}_3)$ follows from
Lemma~\ref{lem:relevant}~(\ref{it:relevant1}).
\end{proof}
  

%% file: omqa.tex

\section{PRPQ Entailment for \SIQfwd and \SOQ}
\label{sec:omqa}



We shall now exploit canonicity of tree decompositions 
in an automata-based decision procedure for query entailment in
\SIQfwd and \SOQ, yielding optimal complexity upper bounds.

Let us fix a (\SIQfwd or \SOQ) KB $\Kmc$ and a PRPQ $\varphi$, and denote with
$\Sigma_\mn{C},\Sigma_\mn{R}^t,\Sigma_{\mn{R}}^{nt}$ the concept
names, transitive role names, and non-transitive role names used in
\Kmc. By Theorem~\ref{thm:tree-like}, if $\varphi$ is not entailed by $\Kmc$,
there exists a counter-model admitting a canonical tree decomposition
of width and outdegree bounded by a constant $N$ single exponential in
$|\Kmc|$. We effectively construct a non-deterministic tree
automaton recognizing such decompositions of
counter-models, and thus reduce query entailment to the emptiness
problem. 

Let us introduce the necessary notions for tree automata.  A
\emph{$k$-ary $\Omega$-labeled tree} is a pair $(T,\tau)$ where $T$ is
a tree each of whose nodes has at most $k$ successors and 
$\tau: T\rightarrow \Omega$ assigns a letter from $\Omega$ to each
node.  
A \emph{non-deterministic tree automaton (NTA)} over $k$-ary
$\Omega$-labeled
trees is a tuple $\Amf=(Q,\Omega,q_0,\Lambda)$, where $Q$ is a finite set of
states, $q_0\in Q$ is the initial state, $\Lambda \subseteq
\bigcup_{i\leq k}(Q\times \Omega \times Q^i)$ is a set of transitions. 
A \emph{run $r$} on a $k$-ary $\Omega$-labeled tree $(T,\tau)$ is a
$Q$-labeled tree $(T,r)$ such that $r(\varepsilon)=q_0$ and, for every
$x\in T$ with successors $x_1,\ldots, x_m$, there is a transition
$(r(x),\tau(x),r(x_1)\cdots r(x_m))\in \Lambda$. As usual,
\Amf \emph{recognizes} the set of all $\Omega$-labeled trees admitting
a run.

Since counter-models have a potentially infinite domain, we 
\emph{encode} tree decompositions of width $N$ using a domain $D$ of $2N$ elements, similar to
what has been done, e.g., in~\cite{GradelW99}.  
Intuitively, if $w$ is a successor node of $v$ in the tree
decomposition, then an element $d$
occurring in (the bag at) $w$
represents a fresh domain element iff $d$ does not occur in $v$.
More precisely, the
alphabet $\Omega$ of the automaton is the set of all pairs $(x,\Imc)$
such that either $x\in\Sigma_\mn{R}$ and $\Imc$ is a
$\Sigma_\mn{C}\cup\{x\}$-interpretation with $\Delta^\Imc\subseteq D$, or
$x=\bot$ and $\Imc$ is a $\Sigma_\mn{C}\cup\Sigma_\mn{R}^{nt}$-interpretation with
$\Delta^\Imc\subseteq D$. 
%

\begin{restatable}{lemma}{lemautomaton}
  \label{lem:automaton}
  Given $\Kk$, $\varphi$, and $N$, one can compute in time
  $O(2^{\mn{poly}(N)})$ an NTA
  recognizing the set of encodings of canonical tree decompositions of width and
  outdegree at most $N$ such that for the underlying interpretation
  $\Jj$ it holds that $\Jj^* \models \Kk$ and $\Jj^* \notmodels
  \varphi$, as well as $\Jj \models \Aa$ and $\Jj\models \Tt$.
\end{restatable}

\begin{proof}
The desired NTA is the intersection of an NTA $\Amf_\Kmc$
recognizing all canonical tree decompositions such that the
underlying interpretation \Jmc satisfies $\Jmc\models\Amc$,
$\Jmc\models \Tmc$, and $\Jmc^*\models\Kmc$ and an NTA
$\Amf_{\neg\varphi}$ recognizing all tree decompositions of
counter-models of $\varphi$.  Since the latter is known
from~\cite[Lemma~6]{GuIbJu-AAAI18}, we concentrate on
$\Amf_\Kmc=(Q,\Omega,q_0,\Lambda)$, working over $N$-ary trees.

Informally, its construction relies on the following ideas:
\textit{(i)} by~\ref{it:canbase3} and~\ref{it:can2}, in every bag
there is at most one $d$ satisfying the condition `$d\in F(u)$ \ldots'
in Condition~\ref{it:can3}; thus, \textit{(ii)} canonicity can be
checked by initially guessing \Mmc and then comparing neighboring
interpretations and remembering the mentioned $d$ in the states;
\textit{(iii)} $\Jmc\models\Amc$ can be verified by looking at labels
of the root and its direct successors; \textit{(iv)} due to canonicity
and the TBox normal form, $\Jmc\models\Tmc$ can be verified by looking
at the current label (this suffices for at-most restrictions over
transitive roles, due to canonicity) and possibly at successor bags
(at-least restrictions, and at-most restrictions over non-transitive
roles); \textit{(v)} $\Jmc^*\models\Tmc$ is a consequence of
$\Jmc\models \Tmc$, by the normal form.

Formally, the set $Q$ contains $q_0$ and all tuples of the shape 
$$\langle (x,\Imc), F, \Mm, \Bmc, \Cmc, e, r,
f\rangle,$$ 
where $(x,\Imc)\in \Omega$, $F\subseteq \Delta^\Imc$,
$\Mm$ is a $\Sigma_\mn{C}\cup\Sigma_\mn{R}^{nt}$-interpretation
with $\Delta^{\Mm}\subseteq D$, 
$\Bmc\subseteq\Amc$, 
%
\Cmc is a set of assertions of the shape
$\qnrgeq{n}{s}{B}(d)$, $\qnrleq{n}{s}{B}(d)$, or $(\exists s.B)(d)$ with 
$d\in D$, $B\in\Sigma_\mn{C}$, $s$ a role in \Tmc, $n\leq N$, and $e,f\in
D\cup\{\varepsilon\}$, $r\in\Sigma_\mn{R}^t$.

In state $q=\langle (x,\Imc),F,\Mm,\Bmc,\Cmc,e,r,f\rangle$
reading symbol $a=(x',\Imc')$, the automaton allows a
transition only in case the following conditions are satisfied:
\begin{itemize}

  \item Conditions~\ref{it:canbase0}--\ref{it:canbase3}
    and~\ref{it:can0}--\ref{it:can4} with $\Imc,\Imc',x,x',F$ taking
    the role of $\Imf(v),\Imf(w),\bago(v),\bago(w),F(v)$,
    respectively, and `$d\in F(u)$ \ldots' in~\ref{it:can3} replaced
    with `$d=f$';

  \item $x'\neq\bot$ and, if $x=\bot$, then $\Imc'\models \Bmc$;
    
  \item either $e\neq \varepsilon$, $e\in \Delta^{\Imc'}$, and $r=
    x'$, or $e=\varepsilon$ and $x=x'$; 

  \item $\Imc'\models C(d)$ for all $C(d)\in \Cmc$ with $d\in
    \Delta^{\Imc'}$ or $C$ of shape $\qnrgeq{n}{s}{B}$ or $\exists
    r^-.B$;

  \item $\Imc'\models \alpha$ for all $\alpha\in \Tmc$ of the form
    $\bigsqcap_i A_i\sqsubseteq \bigsqcap_j B_j$, $A\sqsubseteq
    \qnrleq{n}{r}{B}$, and $A\sqsubseteq \forall r^-.{B}$.

\end{itemize}
%
%
In this case, $\Lambda$ allows all transitions $(q,a,q_1\cdots q_m)$,
$m\leq N$ where each $q_i$ is of shape
$\langle(x',\Imc'),F',\Mm,\emptyset,\Cmc_i,e_i,r_i,f_i\rangle$
with $F'=\Delta^{\Imc'}\setminus\Delta^\Imc$ and:
\begin{itemize}


  \item for each $d\in F'$ and each $r\in \Sigma_\mn{R}^t\setminus\{x'\}$,
    there is a unique~$i$ such that $e_i=d$ and $r_i=r$; conversely,
    if $e_i\neq \varepsilon$ for some $i$, then $e_i\in F'$ and
    $r_i\neq x'$;
  
  \item if $e\neq \varepsilon$, then $f_i=e$, for all $i$;


  \item for all $A\sqsubseteq \exists r^-.B\in \Tmc$ and $d\in
    A^{\Imc'}\cap F'$ such that $d\notin (\exists r^-.B)^{\Imc'}$, we have
    $(\exists r^-.B)(d)\in \Bmc_i$ for some $i$;



  \item for all $A\sqsubseteq \qnrleq{n}{r}{B}\in \Tmc$, $r\in
    \Sigma_\mn{R}^{nt}$, and $d\in A^{\Imc'}\cap F'$, there is a partition $n=
    n_0+\ldots+n_m$, such that $d\in \qnrleq{n_0}{r}{B}^{\Imc'}$, and
    $\qnrleq{n_i}{r}{B}(d)\in \Cmc_i$, for all $i$;

  \item for all $A\sqsubseteq \qnrgeq{n}{s}{B}\in \Tmc$
    and $d\in A^{\Imc'}\cap F'$, there is a partition $n=
    n_0+\ldots+n_m$, such that $d\in \qnrgeq{n_0}{s}{B}^{\Imc'}$ and
     $\qnrgeq{n_i}{s}{B}(d)\in\Cmc_i$, for all $i$ with $n_i>0$.

\end{itemize}
The transitions for $q_0$ are similar, but they additionally nondeterministically initialize
$\Mm$ and check the non-transitive part of the ABox in the root,
see the appendix. Correctness of the automaton is essentially a
consequence of Points~\textit{(i)}--\textit{(v)} mentioned above.
It is routine to verify that $\Amf_\Kmc$ is of the required size and
can be constructed in the required time. 
\end{proof}

Recall that emptiness of NTAs can be checked in polynomial time. Thus,
Lemma~\ref{lem:automaton} together with the bounds on $N$ from
Theorem~\ref{thm:tree-like}, yields a \twoexp upper bound for PRPQ
entailment in \SIQfwd and \SOQ. A matching lower bound is inherited
from positive existential query answering
in~$\mathcal{ALC}$~\cite{CalvaneseEO14}.
%
\begin{theorem}\label{thm:main-omqa}
  PRPQ entailment over $\SIQfwd$ and $\SOQ$ knowledge bases is
  \twoexp-complete.  
\end{theorem}



%% file: fomqa-soq.tex
\section{Finite PEQ Entailment for \SOQ}
\label{sec:fomqa-soq}

The goal of this section is to establish the following result.

\begin{theorem} \label{thm:fomqa-soq}
Finite PEQ entailment over $\SOQ$ knowledge bases is $\twoexp$-complete. 
\end{theorem}

The lower bound follows directly from the result on unrestricted
query entailment for $\ALCO$ \cite{DBLP:conf/kr/NgoOS16}, as the latter logic
enjoys finite controllability. For the upper bound, we carefully adapt
an approach previously used for $\SOF$ \cite{GogaczIM18}, which
relies on the following additional condition imposed on tree-like
counter-models. 

\begin{definition} \label{def:safe}
A canonical tree decomposition is \emph{safe}, if it contains no
infinite downward path such that for each node $w$ in this path,
$\bago(w)$ is the same transitive role name.
\end{definition}

In what follows, by a \emph{counter-witness} we understand a model 
of the ABox and the TBox whose transitive closure is a counter-model.
The approach requires two ingredients: (1)~equivalence of the
existence of a finite counter-model and the existence of a
counter-witness that admits a safe canonical tree decomposition, and
(2)~effective regularity of the set of safe canonical
tree decompositions (of given width and outdegree) of counter-witnesses.
For (2), observe that safety can be easily checked
by an automaton with B\"uchi acceptance condition~\cite{2001automata}
and the number of states quadratic in the number of transitive role names
in $\Kk$: on each path the automaton remembers the role names
associated with two most recently visited nodes; the state is
accepting unless they are the same transitive role name. The product
of this automaton and the one constructed in the previous section
recognizes the desired language.
Assuming (1) is also available, the upper bound follows like for the
unrestricted case: the algorithm builds the automaton and tests its emptiness.

The reminder of this section provides (1). One
implication is obtained via the following observation. 

\begin{restatable}{lemma}{lemmasoqsafe}
\label{lem:soq-safe}
If $\Ii$ is a finite interpretation of a $\SOQ$ KB, then the unravelling procedure
from the proof of Theorem~\ref{thm:tree-like} yields a safe tree decomposition. 
\end{restatable}

To prove the converse implication we begin from a carefully chosen
counter-witness with a safe canonical tree decomposition. It is well
known that each regular set of trees contains a \emph{regular tree},
i.e., a tree with finitely many non-isomorphic
subtrees. Hence, if there is a counter-witness with a safe canonical
tree decomposition, there is also one with a regular safe canonical
tree decomposition. Let $\Tmf = (T,\Imf)$ be such a tree decomposition
of some counter-witness $\Ii$, and let $\Mm$ be the interpretation
guaranteed by Definition~\ref{def:canonical}. 

Let us restructure $\Tmf$ by iteratively merging neighboring nodes
associated to the same transitive role name: pick a node $v$ with a
child $w$ such that $\bago(v)=\bago(w) \in \troles$, redefine
$\Imf(v)$ as $\Imf(v)\cup\Imf(w)$, remove $w$ from $\Tmf$, and promote
all children of $w$ to children of $v$. As a result we obtain a
canonical tree decomposition $\Smf= (S,\Imf)$ of $\Ii$. By
construction,  $\Smf$ is \emph{strongly canonical}: no neighboring
nodes in $\Smf$ are associated with the same transitive role name.
Hence, for each node $w$ with
parent $v\neq \varepsilon$, $\Delta_v \cap \Delta_w \setminus
\Delta^\Mm= \{d_w\} $ for some $d_w \in F(v)$.

Each regular safe tree decomposition has bounded length of downward
paths of nodes associated with the same transitive role
name. Consequently, the restructuring above keeps the outdegree and
the width bounded.

\begin{restatable}{lemma}{lemmasoqwidth}
\label{lem:soq-width}
$\Smf$ has bounded degree and width.
\end{restatable}

We can now easily turn $\Ii^*$ into a finite model of $\Kk$. Suppose
that on each path of $\Smf$, we fix a node $v$ and its ancestor $u$
(neither $\varepsilon$ nor a child of $\varepsilon$) such that
$\bago(v)=\bago(u)$, $\Imf(v) \simeq \Imf(u)$, and the witnessing
isomorphism $h$ maps $d_v$ to $d_u$ and is identity over
$\Delta^\Mm $. Suppose also that for each element in $\Delta^\Mm$, all
witnesses required by at-least restrictions can be found among
elements of $\Delta^\Mm$ and elements that do not occur in the
subtrees of $\Smf$ rooted at the chosen nodes $v$. Note that for each
path we can find such a pair of nodes, because the sizes of the bags
are bounded in $\Smf$. 
We shall modify $\Ii$ by removing parts of it and redirecting edges
previously leading to the removed parts. Pick any path such that the
corresponding $d_v$ has not been processed yet and has not been
removed. Remove from $\Ii$ the union of $\Delta_w$ with $w$ ranging
over descendents of $v$ (including $v$), keeping only $\Delta^\Mm$ and
$d_v$.
Replace each $\bago(v)$-edge leading from $d_v$ to a removed element
$e\in\Delta_v$, with an $\bago(v)$-edge leading from $d_v$ to $h(e)$.
Repeat until no such path exists. The resulting
interpretation $\Jj$ is obviously finite. Checking correctness is
routine.

\begin{restatable}{lemma}{lemmasoqmodel}
\label{lem:soq-model}
$\Jj^* \models \Kk\,$.
\end{restatable} 



To ensure that $\Jj^* \notmodels \varphi$ we need to choose the nodes
$v$ and $u$ more carefully. Relying on $\varphi$ being a PEQ,
 not an arbitrary PRPQ,
we apply the colored blocking principle
\cite{GogaczIM18}: to keep $u$ and $v$ sufficiently similar and
sufficiently far apart, we look at their neighborhoods
of sufficiently large radius and use additional coloring to
distinguish elements within each neighborhood (see Appendix for details).


%% file: finsat-siq.tex
\section{IQ Entailment}\label{sec:finsat-siq}

We also get the following results on IQ entailment.

\begin{theorem}\label{thm:IQSIQ}
  Finite and unrestricted IQ entailment is $\textnormal{\sc co}$\NExpTime-complete over
  \SOQ KBs.
  Unrestricted IQ entailment over \SIQfwd KBs and
  finite IQ entailment over \SIQfwd KBs restricted to a single
  transitive role and no non-transitive roles is
  \twoexp-complete.
%
\end{theorem}
\begin{proof}

As is well-known, (finite) IQ entailment reduces to the
complement of (finite) KB satisfiability; we focus on
the latter. For \SOQ, we simply use the facts
that (finite) KB satisfiability for \SOQ is
\NExpTime-complete~\cite{GuIbJu-AAAI17} and that the lower bound holds
in the finite~\cite{KazakovP09}.

The \twoexp upper bound for
unrestricted KB satisfiability follows from Theorem~\ref{thm:main-omqa}, and 
for finite satisfiability in the above fragment of \SIQfwd  follows
from Theorem~\ref{thm:main-omqa} and a recent result
by~\citeauthor{BKieronski18}~[\citeyear{BKieronski18}], implying that
satisfiability and finite satisfiability coincide for this fragment of
$\SIQfwd$. This approach cannot be generalized to full \SIQfwd since
\SIQfwd lacks the finite model property.

We next show that these upper bounds are tight by reducing the word problem for
$2^n$-space bounded \emph{alternating Turing machines (ATMs)},
which is known to be \twoexp-hard~\cite{chandraAlternation1981}.  An
ATM $M=(Q,\Theta,\Gamma,q_0,\Delta)$ consists of the set  $Q$ of
states partitioned into \emph{existential states}~$Q_{\exists}$ and \emph{universal
states}~$Q_{\forall}$, the input alphabet $\Theta$, the 
tape alphabet $\Gamma$,  the \emph{starting state} $q_0\in Q_{\exists}$, and
the \emph{transition relation} $\Delta$.  Without loss of
generality, we assume that 
each of $M$'s configurations has exactly two successor configurations,
universal and existential states alternate, and $M$ accepts a word iff
there is an infinite (alternating) run.

Given $M,w$, we construct in polynomial time a knowledge base
$\Kmc=(\Tmc,\{I(a_0)\})$ using a single transitive role name $r$ such
that $M$ accepts $w$ iff $\Kmc$ is satisfiable. We represent
configurations of size $2^n$ in the leaves of binary trees of depth
$n$. To this end, we use concept names $X_0,\ldots, X_{n-1}$
(representing bits of an exponential counter) and $L_0, \ldots, L_{n}$
(for the levels of the tree), and include the following CIs for all
$i,j$ with $0\leq i<n$ and $i<j\leq n$:
\begin{align*}
  L_i &\sqsubseteq  \exists r. (X_i\sqcap L_{i+1}) \sqcap \exists r.
  (\neg X_i\sqcap L_{i+1}) \\
  %
  %
  L_{i+1} \sqcap X_i & \sqsubseteq \forall r.(L_j \rightarrow X_i ) \\
  L_{i+1} \sqcap \neg X_i & \sqsubseteq \forall r.(L_j \rightarrow
  \neg X_i)
  %
  %
\end{align*}
It should be clear that every model of $L_0$ and these CIs contains a
full binary tree of depth $n$ such that the \emph{leaves}, that is,
the elements satisfying $L_n$, correspond to the numbers
$0,\ldots,2^n-1$ in the natural way, via concept names $X_i$.
Let $\Sigma=\Gamma\cup(Q\times \Gamma)$ be the set of possible labels of a
cell in $M$'s computation, and introduce concept names $C^x_\sigma$
for every $\sigma\in \Sigma$, $x\in \{l,h,r\}$. Every leaf with number $i$
is labeled with three concepts $C^l_{\sigma_1}, C^h_{\sigma_2},
C^r_{\sigma_3}$ representing the cells $i-1,i,i+1$ of a
configuration using the CI:
\begin{align*}
  L_n & \sqsubseteq \textstyle\bigsqcap_{x\in\{l,h,r\}}\bigsqcup_{\sigma\in
  \Sigma} \big( C_\sigma^x \sqcap \bigsqcap_{\sigma'\neq \sigma}\neg
  C_{\sigma'}^x\big) 
\end{align*}
We use these trees as follows. The concept name $I$ enforces
a skeleton structure modeling an alternating
computation using the following CIs, for $i\in\{1,2\}$: 
\begin{align*}
  I \sqsubseteq A_{\exists}^1\quad\quad 
  A_{\exists}^i & \sqsubseteq L_0\sqcap 
  \exists r^-.(B_{\exists}\sqcap \exists r.A_{\forall}) \\
  %
  A_{\forall} & \sqsubseteq L_0\sqcap 
  \exists r^-.(B_{\forall}^i\sqcap \exists r.A_{\exists}^i)
  %
  %
\end{align*}
Thus, every model of $I$ contains the following structure, where every triangle
represents one of the described trees, and $A_\forall$
($A_{\exists}^i$) marks universal (existential) configurations:
 
\begin{tikzpicture}
  \node[shape=circle,fill,draw,inner sep=1pt,label=above:{\scriptsize
  $I,A_{\exists}^1,L_0$}] (r1) at (0,0) {};
  \node[shape=circle,fill,draw,inner
  sep=1pt,label=above:{\scriptsize$B_\exists$},label=below:{\scriptsize{$d$}}] (r2) at (1,0) {};
  \node[shape=circle,fill,draw,inner
  sep=1pt,label=above:{\scriptsize\begin{tabular}{c}$A_\forall,$\\$L_0$\end{tabular}}] (r3) at (2,0) {};
  \node[shape=circle,fill,draw,inner
  sep=1pt,label=above:{\scriptsize$B_{\forall}^1$}] (r4) at (2.7,0.7) {};
  \node[shape=circle,fill,draw,inner
  sep=1pt,label=above:{\scriptsize$B_{\forall}^2$}] (r5) at (2.7,-.7) {};
  \node[shape=circle,fill,draw,inner sep=1pt,label=above:{\scriptsize
  $A_{\exists}^1,L_0$}] (r6) at (3.7,0.7) {};
  \node[shape=circle,fill,draw,inner sep=1pt,label=above:{\scriptsize
  $A_{\exists}^2,L_0$}] (r7) at (3.7,-.7) {};
  \node[shape=circle,fill,draw,inner sep=1pt,label=above:{\scriptsize
  $B_{\exists}$}] (r8) at (4.7,0.7) {};
  \node[shape=circle,fill,draw,inner sep=1pt,label=above:{\scriptsize
  $A_{\forall},L_0$}] (r9) at (5.7,0.7) {};
   \node[shape=circle,fill,draw,inner sep=1pt,label=above:{\scriptsize
  $B_{\exists}$}] (r10) at (4.7,-.7) {};
  \node[shape=circle,fill,draw,inner sep=1pt,label=above:{\scriptsize
  $A_{\forall},L_0$}] (r11) at (5.7,-.7) {};
  \node at (6.6,.7) {$\ldots$};
  \node at (6.6,-.7) {$\ldots$};
  \node at (0,-.5) {$T_1$};
  \node at (2,-.5) {$T_2$};
   
  \draw[->,>=stealth] (r2) -- (r1);
  \draw[->,>=stealth] (r2) -- (r3);
  \draw[->,>=stealth] (r4) -- (r3);
  \draw[->,>=stealth] (r4) -- (r6);
  \draw[->,>=stealth] (r5) -- (r3);
  \draw[->,>=stealth] (r5) -- (r7);
  \draw[->,>=stealth] (r8) -- (r6);
  \draw[->,>=stealth] (r8) -- (r9);
  \draw[->,>=stealth] (r10) -- (r7);
  \draw[->,>=stealth] (r10) -- (r11);
  \draw[<-,>=stealth] (r11) -- (6.2,-.9);
  \draw[<-,>=stealth] (r11) -- (6.2,-.5);
  \draw[<-,>=stealth] (r9) -- (6.2,.9);
  \draw[<-,>=stealth] (r9) -- (6.2,.5);

  \draw (0,0) -- (0.4,-0.7) -- (-0.4,-.7) -- cycle; 
  \draw (2,0) -- (2.4,-0.7) -- (1.6,-.7) -- cycle; 
  \draw (3.7,-.7) -- (4.1,-1.4) -- (3.3,-1.4) -- cycle; 
  \draw (3.7,0.7) -- (4.1,0) -- (3.3,0) -- cycle; 
  \draw (5.7,-.7) -- (6.1,-1.4) -- (5.3,-1.4) -- cycle; 
  \draw (5.7,0.7) -- (6.1,0) -- (5.3,0) -- cycle; 
 
\end{tikzpicture}

\noindent It remains to ensure that \textit{(i)} the leaf labeling in
every tree is actually a configuration,~\textit{(ii)} neighboring
trees describe successor configurations, and~\textit{(iii)} the first
tree is labeled with the initial configuration. We concentrate
on~\textit{(ii)}, as~\textit{(i)} is similar and~\textit{(iii)} is
straightforward. We illustrate the idea on $T_1$ and $T_2$ in the
figure. In~$T_1$, we enforce in every leaf an
$r$-successor satisfying the label of that cell in the successor
configuration (computable from the $C^x_\sigma$).  In~$T_2$, we
enforce in every leaf
an $r$-successor with the current label $C^h_\sigma$. Both in $T_1$
and $T_2$, these additional elements satisfy a fresh concept name $S$
and have the same counter value as in the leaves.  
Observe that, by transitivity, all $2\cdot 2^n$ created nodes are `visible' from
$d$ satisfying $B_\exists$ in the figure. By including the CI
$B_\exists\sqsubseteq \qnrleq{2^n}{r}{S}$, $S$-elements with the same
counter value from $T_1$ and $T_2$ are forced to identify, thus
achieving the desired synchronization.
Having~\textit{(i)}--\textit{(iii)} in place, it is routine to show
that \Kmc is satisfiable iff $M$ accepts $w$. 
The lower bound applies to finite satisfiability, since $\Kmc$ is
satisfiable iff it is finitely satisfiable. 
\end{proof}

 


%% file: conclusions.tex
\section{Outlook}
This paper makes a step towards a complete picture of query
entailment in DLs with number restrictions on transitive roles.
%
There are several natural next steps involving finite entailment.
The first is to cover full \SIQfwd. A more challenging goal is to go
beyond instance queries: an immediate obstacle is that the natural
safety condition for \SIQfwd does not guarantee strongly canonical
decompositions. Covering full PRPQs even just for \SQ seems to require
generalizing the colored blocking principle, or finding an entirely
different tool.

\section*{Acknowledgments}
 
Gogacz and Murlak were funded by  Poland's National Science Centre
grant 2018/30/E/ST6/00042, Ib\'a\~nez-Garc\'ia by ERC Starting grant
637277 FLEXILOG, and Jung by ERC Consolidator grant 647289 CODA. This
work was also supported by the OeAD WTZ project PL 15/2017.

%% file: additional-notation.tex
\section{Additional Preliminaries}

We define the semantics of PRPQs as follows:
\begin{itemize}

  \item $\Imc,\pi\models \psi_1\vee\psi_2$ iff $\Imc,\pi\models \psi_1$ or
    $\Imc,\pi\models \psi_2$;

  \item $\Imc,\pi\models \psi_1\wedge\psi_2$ iff $\Imc,\pi\models
    \psi_1$ and $\Imc,\pi\models \psi_2$;

  \item $\Imc,\pi\models \Emc(t,t')$ iff
    $(\pi(t),\pi(t'))\in\Emc^\Imc$, with $\Emc^\Imc$ defined as
    follows: 
    \begin{align*}
      (r^-)^\Imc & = \{ (e,d) \mid (d,e)\in r^\Imc \} \\
      (A?)^\Imc &= \{(d,d)\mid d\in A^\Imc\} \\ 
      (\Emc^*)^\Imc & = (\Emc^\Imc)^* \\
      (\Emc\cup \Emc')^\Imc & = \Emc^\Imc\cup {\Emc'}^{\Imc} \\
      (\Emc\circ \Emc')^\Imc & = \Emc^\Imc\circ {\Emc'}^{\Imc}
     \end{align*}
\end{itemize}

\medskip
A \emph{tree} is a prefix-closed subset $T\subseteq
(\mathbbm{N}\setminus\{0\})^*$. A node $w\in T$ is a \emph{successor} of
$v\in T$ and $v$ is a \emph{predecessor} of $w$ if $w=v\cdot i$ for some
$i\in \mathbbm{N}$. We say  that the node $\varepsilon$ is the \emph{root} of $T$.

\medskip
A \emph{(non-deterministic) tree automaton} with \emph{B\"uchi acceptance condition}
is a tree automaton enriched with a set of accepting states $F
\subseteq Q$. A run of such an automaton is considered accepting if
on each branch accepting states occur infinitely often.

\subsection*{Normal Form}

As stated in the main part of the paper we assume normalized  KBs, such that each CIs in the TBox takes one of the following forms:
\begin{align*}
  \textstyle\bigsqcap_i A_i \sqsubseteq \bigsqcup_j B_j,
  \quad A \sqsubseteq \forall r^-. B, \quad A\sqsubseteq \exists
  r^-.B, \\ 
  \quad A \sqsubseteq \qnrleq n s B,\quad A \sqsubseteq \qnrgeq n s B,
\end{align*}
where $A,A_i,B,B_j$ are concept names or nominals, $r\in\roles$, $s$
is a non-transitive role or a transitive role name, and empty
disjunction and conjunction are equivalent to $\bot$ and $\top$,
respectively. This can be assumed w.l.o.g. since every \SIQfwd or \SOQ
TBox can be transformed into a normalized one by extending its
signature with an appropriate number (linear on the size of the TBox)
of fresh concept names.  We further assume that for every at-most
restriction \qnrleq n s B, \Tmc also contains the following
\begin{equation} \label{eq:NF}
     A \sqsubseteq \qnrleq n s B \,, \;
     A'  \sqsubseteq \qnrgeq {n+1} s B  \,,\;
   \top \sqsubseteq A \sqcup A'
\end{equation}
 with $A$ a concept name, not occurring on the left-hand-side of any
 other CI.  We make an analogous assumption for at-least
 restrictions.
 

%% file: tree-likeMP-appendix.tex
\section{Additional Proofs for Section~\ref{sec:tree-likeMP}}
\label{sec:tree-likeMP-appendix}

\lemrelevant*

\begin{proof} 
  Point~1 is a consequence of the definition of $\rel_r^\Imc(d)$.

  For Point~2, let us denote with $\mn{atm}(\Tmc)$ the set of all
  at-most restrictions occurring in $\Tmc$. We
  say that $e\in \Delta^\Imc$ is \emph{directly
  relevant for $d\in\Delta^\Imc$} if there is some
  $\qnrleq{n}{r}{B}\in \Tmc$, such that $d\in \qnrleq n r
  B^{\Imc^*}$, $e\in B^\Imc$, and $(d,e)\in r^{\Imc^*}$. We further
  denote with $X^\Imc_r(d)$ the smallest set that contains $d$ and is
  closed under direct relevant elements. Because $\Ii\models \Tmc$, we have
  $$\rel_r^\Imc(d) \subseteq \bigcup_{e\in X_{r}^\Imc(d)}Q_r^\Imc(e).$$ It thus suffices to prove that the size
  of $X_{r}^\Imc$ is bounded by $2^{\mn{poly}(|\Tmc|)}$.

  To see this, consider the directed tree
  $(V,E)$ with $V=X_{r}^\Imc(d)$ and $E$ is defined as follows. Start with
  setting $E$ the set of all $(d,e)$ such that $e$ is directly
  relevant for $d$ and apply the following step exhaustively. 
  \begin{itemize}

    \item[$(\ast)$] Choose leaf $e\in V$ and add, for all $f\in
      \Delta^\Imc$ directly relevant for $e$, but not for any of
      $e$'s predecessors an edge $(e,f)$ to $E$.
      
   \end{itemize}
   By definition of $V$ and direct relevance, $(V,E)$ is a connected tree. 
   Now, consider the labelling
   $\ell:V\to 2^{\mn{atm}(\Tmc)}$ given by
   $$\ell(e)=\{C\mid e\in {\qnrleq n r C}^{\Imc}, \qnrleq n r C
   \in\mn{atm}(\Tmc)\}.$$ 
   Let $(e,f)\in E$. By construction, we have
   \begin{itemize}

     \item[--] $\ell(e)\subseteq \ell(f)$ if $f$ is a leaf in
       $(V,E)$, and

     \item[--] $\ell(e)\subsetneq \ell(f)$ if $f$ is an inner node in $(V,E)$.

   \end{itemize}
   Thus, the depth of the tree $(V,E)$ is bounded by $|\Tmc|$.
   Observe moreover that also the outdegree of $(V,E)$ is bounded
   exponentially in $\Tmc$ by definition of direct relevance.
   Overall, we get that the size of $V=X_{r}^\Imc(d)$ is bounded by an exponential in
   \Tmc.
\end{proof}

\thmtreelike*

The construction of the decomposition $\Tmf$ described in the body of
the paper directly ensures that $\Tmf$ is a canonical decomposition of
width and outdegree appropriately bounded.  It remains to check that
the interpretation $\Jj$ underlying the decomposition $\Tmf$ satisfies
the conditions required in the statement of the theorem. 

\begin{claim}\label{cl:model}
 $\Jmc \models (\Tmc, \Amc)$ 
\end{claim}

\begin{proof}
During the proof,  we will repeatedly use the following fact, a consequence of the  existence of a homomorphism from  $\Jmc$ to \Imc. 
\begin{itemize}
\item[($\dagger$)] The unary type of each element $d'$ in \Jmc coincides with that of its original $d$ in \Imc. 
\end{itemize}

All the assertions in $\Amc$ are satisfied in $\Jmc$ since
the root bag in the tree decomposition of  $\Jmc$ contains all the
individuals in \Amc, as well as  all the non-transitive edges
involving those elements,  and  because all edges of transitive roles
among individuals are added by $(\mathbf{R}_0)$.  

If $\Kmc$ is a \SOQ KB, the construction of $\Jmc$ ensures that each
nominal is interpreted by a singleton. Indeed, since every bag in the
tree decomposition contains the set of elements $\Delta$,  consisting
of all the nominals and their relevant $r$-successors for every
transitive role $r$, the definition of the root bag and rules
$(\mathbf{R}_0)$ -- $(\mathbf{R}_3)$  ensure that for every role name  
edges between elements in $\Imc$ and $\Delta$ are faithfully
replicated in  $\Jmc$.   
 
Next, we need to show that for every  CI $C \sqsubseteq D  \in \Tmc$
it holds that $C^\Jmc \subseteq D^\Jmc$. For the case where $D$ is of
the form $\bigsqcup_j B_j$,  this follows directly from ($\dagger$). 
For the case where $D$ is a universal, existential, at-most or
at-least restriction over a non-transitive role, the statement holds
by construction because of $(\mathbf{R}_1)$. It thus remains to
consider restrictions involving transitive roles. Assume that $r$ is a
transitive role name and let $f' \in C^\Jmc$.  

If $D =\exists r^-.B$ or $D = \qnrgeq n r B$ for some concept name
$B$, then $(\mathbf{R}_3)$ ensures that $f' \in D^\Imc$, whereas for
the case where $D = \forall r^-. B$, this follows since $B$ is a
concept name and because of ($\dagger$).  Indeed, assume $(e', f') \in
r^\Jmc$, and let $e$ be the original of $e'$.  By construction, $\Jmc$
can be mapped homomorphically to $\Imc$, and thus we have $(e,f) \in
r^\Imc$. Because the unary types of $f$ and $f'$ coincide, and $\Imc
\models \Kmc$, it holds that $e \in B^\Imc$ and therefore $e' \in
B^\Jmc$.

The case where $D = \qnrleq n r B$ is slightly more subtle.  Let $f$
be the original of $f'$ in \Imc. We shall need a subclaim that during
the construction of \Jmc exactly one copy of each element from
$\rel^\Ii_r(f)$ was introduced among $r$-successors of $f'$. 
Let $w$ be the bag closest to the root bag containing $f'$ and such
that $\bago(w)= r$. Let us see that $w$ contains 
exactly one copy of each element from $\rel^\Ii_r(f)$. Indeed, $w$
could have only been introduced by $(\mathbf{R}_0)$, 
$(\mathbf{R}_2)$ or $(\mathbf{R}_3)$. 

If $w$ was introduced by $(\mathbf{R}_2)$, using
point~\ref{it:relevant1} of Lemma~\ref{lem:relevant} we see
immediately that $\rel^\Ii_r(f) \subseteq \rel^\Ii_r(d)$, where $d$ is
like in the formulation of rule $(\mathbf{R}_2)$, and so it is clear
that $\Imf(w)$ contains exactly one copy of each element from
$\rel^\Ii_r(f)$.

For $w$ introduced by $(\mathbf{R}_0)$, the argument is analogous.

Now assume $w$ was introduced by $(\mathbf{R}_3)$. Let $v$ be the
parent of $w$ and let $d'$ and $d$ be like in the formulation of rule
$(\mathbf{R}_3)$. Then, $f'$ is a copy of an element $f$ from
$\rel^\Ii_r(e)$, where $e$ is an $r$-successor
(or $r^-$-successor) of $d$ and $e \notin \rel^\Ii_r(d)\cup\Delta$. By
point~\ref{it:relevant1} of Lemma~\ref{lem:relevant}, $\rel^\Ii_r(f)
\subseteq  \rel^\Ii_r(e)$. Hence, the rule adds exactly one copy of each
element of $\rel^\Ii_r(f)$: either a fresh copy, or the copy inherited
from $v$ if the element belongs to $\rel^\Ii_r(d) \cup \Delta$.
     
Every other possible $r$-successor of $f'$ can only be added by a
subsequent application of $(\mathbf{R}_3)$.  The definition of this
rule ensures that each fresh element added in this way is a copy of an
$r$-successor of $f$ that does not belong to $\rel^\Ii_r(f)$. This
completes the proof of the subclaim.

Now, assume that $f' \in C^\Jj$. Because $\Jj$ maps homomorphically
into $\Ii$, if $g'$ is an $r$-successor of $f'$ that belongs to
$B^\Jj$, then its original $g$ is an $r$-successor of $f$ that belongs to
$B^\Ii$. We also know that $f\in C^\Ii$ and
because $\Ii\models \Tt$ we know that there is at most $n$ such
elements $g$ and they all belong to $\rel^\Ii_r(f)$. But elements of
$\rel^\Ii_r(f)$ are copied exactly once among $r$-successors of $f'$,
so the number of possible $g'$ is also bounded by $n$, and we are
done. 
\end{proof}

\begin{claim}
  $\Jmc^* \models \Kmc$
\end{claim}

\begin{proof} 
Clearly, by definition of $\Jj^*$, every transitive role is
interpreted as a transitive relation. 

Now, we observe that the homomorphism from $\Jmc$ to $\Imc$ can be
naturally extended to  a homomorphism from $\Jmc^*$ to $\Imc$. Thus,
($\dagger$) applies also to the unary types of elements in $\Jmc^*$.
Therefore, $\Jmc \models (\Tmc, \Amc)$ implies that $\Jmc^* \models
\Amc$, and that $\Jmc^*$ satisfies every concept inclusion $C
\sqsubseteq D$, for every $D$ of the form $\bigsqcap_j B_j$, $\exists
r^-. B$ with $r$ a role name, and $\qnrgeq n s B$
with $s$ a non-transitive role or a transitive role name. It remains
to deal with universal and at-most restrictions. For
$D=\forall r^-. B$, the argument is the same as for $\Jj$.

To show that $\Jmc^*$ satisfies CIs of the form $C \sqsubseteq \qnrleq
n r B$, with $r$ a transitive role name, we use the subclaim used in the proof for  $\Jmc$. 
%
%
Let $f' \in C^{\Jmc^*}$ and $w$ be the bag closest to the root such
that $f'\in F(w)$ and $\bago(w)= r$.  We will show that no fresh
$r$-successors of $f'$ (in $\Jj^*$) violating the at-most restriction
are added as a result of multiple applications of rule
$(\mathbf{R}_3)$. More precisely, we show that for every $g' \in
B^{\Jmc^*}$, if $(f',g') \in r^{\Jmc^*}$, then $g' \in \Delta_w$.

Towards a contradiction, suppose that at some point among those
successors one fresh element $g'$ is added, such that $g' \in
B^\Jmc$. Then there exists a path of bags $w= w_0, w_1, w_2 \dots,
w_k$ with $k>0$,
corresponding to $r$-successors $f'= e'_0, e'_1,e'_2, \dots, e'_k=g'$
and their originals $f= e_0, e_1, e_2, \dots, e_k= g$, such that $e'_k
\in F(w_k)$.
Since $\Kmc$ is normalized, \eqref{eq:NF} ensures that there is a
concept name $A$ that is equivalent to $\qnrleq n r B $. Because
$\Imc \models \Kmc$, we have $f \in A^\Imc$; it then follows by
$(\dagger)$ that $f' \in A^\Jmc$. Furthermore, $(f, e_{k-1}) \in
r^\Imc$, implies $e_{k-1} \in {\qnrleq n r B}^\Imc$ and therefore
$e_{k-1} \in A^\Imc$. Since $e_k \in B^\Imc$, it then follows that
$e_k \in \rel^\Ii_r(e_{k-1})$. However, this contradicts the
assumption that $e'_k \in F(w_k)$, because by the subclaim, there is
exactly one copy of every relevant successor of $e_{k-1}$ in
$\Delta_{k-1}$, and the definition of $(\mathbf{R}_3)$ prevents the
introduction of fresh copies of these elements.
%
%
%
%
\end{proof}


%% file: omqa-appendix.tex
\section{Missing proofs for Section~\ref{sec:omqa}}

\subsection{Encoding Tree Decompositions}

We provide missing details on the encoding.  Let $(T,\tau)$ be a
$\Omega$-labeled tree with $\Omega$ defined as in the main part. For
convenience, we use $\Imc_w$ and $r_w$ to refer to the single
components of $\tau$ in a node $w$, that is, $\tau(w)=(r_w,\Imc_w)$.
Given an element $d\in \Delta$, we say that $v,w\in T$ are
\emph{$d$-connected} iff $d\in\Delta^{\Imc_u}$ for all $u$ on the
unique shortest path from $v$ to~$w$. In case $d\in \Delta^{\Imc_w}$,
we use $[w]_d$ to denote the set of all $v$ which are $d$-connected
to~$w$.  We call $(T,\tau)$ \emph{\Amc-consistent} if
$\mn{ind}(\Amc)\subseteq\Delta^{\Imc_\varepsilon}$.  An
\Amc-consistent $\Omega$-labeled tree $(T,\tau)$ \emph{represents} a
pair $(T,\bagz)$ where the interpretations $\bagz(w)$ are defined by
taking, for all $w\in T$: 
\begin{align*}
  \Delta_w & = \{[w]_d\mid d\in \Delta^{\Imc_w}\},\\
  A^{\Imf(w)} & = \{[w]_d\mid d\in A^{\Imc_w}\}, \\
  r^{\Imf(w)} & = \{ ([w]_{d},[w]_{e})\mid (d,e)\in r^{\Imc_w}\},
\end{align*} 
for all concept names $A$ and role names $r$ occurring in \Kmc. We
further associate an interpretation $\Imc_{(T,\tau)}$ to every
consistent $\Omega$-labeled tree by taking
$\Imc_{(T,\tau)}=\bigcup_{w\in T}\Imf(w)$ and interpret individual
names $a\in \mn{ind}(\Amc)$ by taking $a^{\Imc_{(T,\tau)}} =
[\varepsilon]_a$. Note that this is well-defined due to
\Amc-consistency. It can be easily verified that $(T,\bagz)$ is a tree
decomposition of $\Imc_{(T,\tau)}$. Conversely, given some
interpretation \Imc and a tree decomposition $(T,\Imf)$ of \Imc of
width $K$, one can construct a \Amc-consistent $(T,\tau)$ such that
$\Imc_{(T,\tau)}$ is isomorphic to \Imc, based on the size $2K$ of
$\Delta$~\cite{GradelW99}. Note that in both cases the outdegree is
preserved so that it suffices to consider $N$-ary trees throughout.

\lemautomaton*

We provide the missing transitions for the initial state $q_0$ on
input symbol $(x,\Imc)$. The automaton allows transitions only in case
$x=\bot$, $\mn{ind}(\Amc)\subseteq \Delta^\Imc$, and
$\Imc\models\Amc^{nt}$ where $\Amc^{nt}$ is obtained from \Amc by dropping all
assertions $r(a,b)$ with $r$ a transitive role.
In such a case, $\Lambda$ contains all transitions
$(q_0,(x,\Imc),q_1\cdots q_m)$, $m\leq N$ where 
$q_i$ is of shape $\langle (x,\Imc), F', \Mmc, \Bmc_i,\Cmc_i,
e_i,r_i,f_i\rangle$ with $F'=\Delta^{\Imc}$ and $\Mmc$ some
$\Sigma_{\mn{C}}\cup\Sigma_{\mn{R}}^{nt}$-interpretation with
$\Delta^{\Mmc}\subseteq \Delta^\Imc$ such that the five
conditions from the main part are satisfied and additionally
$\Amc=\Amc^{nt}\cup\bigcup_{i}\Bmc_i$.


%% file: fomqa-soq-appendix.tex

\section{Missing proofs for Section~\ref{sec:fomqa-soq}}

\lemmasoqsafe*

\begin{proof}
Let $w$ be a child of a node $v$ with $\bago(w)=\bago(v) = r$ for some
$r\in\troles$. Then, $w$ was added to the tree decomposition by
applying the rule $(\mathbf{R}_3)$. Let $e$ and $d$ be as in
$(\mathbf{R}_3)$. Because we are working with a $\SOQ$ 
KB,  $e$ is an $r$-successor of $d$. We can also conclude
that there is no $r$-path from $e$ to $d$, because otherwise we would
have $e\in Q_r^\Ii(d) \subseteq \rel^\Ii_r(d)$, which is explictly excluded in
$(\mathbf{R}_3)$. It follows that the length of any path of nodes $u$ with
$\bago(u)=r$ is bounded by the number of $r$-clusters in $\Ii$, which is
finite.
\end{proof}

\lemmasoqwidth*

\begin{proof}
$\Tmf$ is regular, so it has at most $p$
non-isomorphic subtrees for some $p\in\mathbb{N}$. 

Consider a downward path of nodes in $\Tmf$ associated with the same
transitive role $r$. If this path is longer then $p$, the subtrees
rooted at some two nodes on this path are isomorphic. Because one of
them is a proper subtree of the other, it follows immediately that
$\Tmf$ contains an infinite downward path of nodes associated with
$r$, which contradicts safety. Consequently, such a path has length at
most $p$.

Consider a connected subset of nodes of $\Tmf$ associated
with a transitive role name $r$. It is necessarily a subtree of
$\Tmf$. By the argument above, the height of this tree is at most
$p$. Because $\Tmf$ has bounded outdegree, the size of the subtree 
is bounded. 

It follows that each interpretation assigned to a node of $\Smf$ is
the union of a bounded number of interpretations of bounded
size. Hence, $\Smf$ has bounded width.

Similarly, the outdegree of any node of $\Smf$ is the sum
of the outdegrees of a bounded number of nodes of $\Tmf$. Hence,
$\Smf$ has bounded outdegree.
\end{proof}






\subsection{Coloured Blocking Principle}

We first recall (and adapt) key definitions and technical results
underlying colored blocking \citeauthor{GogaczIM18}. The difference
wrt. to the original is that the set of elements that need to be
excluded from the interpretation to make it well-behaved is now
$\Delta^\Mm$, not just $\Nomi(\Kk)$. This does not affect the cited
results.  The whole development is entirely independent of the
knowledge base, and works for any finite set of excluded elements. The
only property of nominals that is ever used is that they are preserved
by homomorphisms. Thus, in what follows all homomorphisms (and
consequently all isomorphisms) are assumed to be identity over
$\Delta^\Mm$. We shall write $\Ii \setminus \Delta^\Mm$ for the
interpretation $\Ii$ restricted to the domain $\Delta^\Ii \setminus
\Delta^\Mm$.

\begin{definition} \label{def:neighbourhood}
For $d \in\Delta^\Ii \setminus \Delta^\Mm$, the
\emph{$n$-neighbourhood} $N_n^{\Ii}(d)$ is the subinterpretation of
$\Ii$ induced by $\Delta^\Mm$ and all elements $e \in \Delta^\Ii
\setminus \Delta^\Mm$ within distance $n$ from $d$ in $\Ii \setminus
\Delta^\Mm$, enriched with a fresh concept interpreted as $\{d\}$. For
$a \in \Delta^\Mm$, $N_n^{\Ii}(a)$ is the subinterpretation induced by
$\Delta^\Mm$, enriched similarly.
\end{definition}

\begin{definition} \label{def:coloring}
A \emph{coloring with $k$ colors} of an interpretation $\Ii$ is any
interpretation $\Ii'$ that enriches $\Ii$ with $k$ fresh concept names
$B_1, \dots, B_k$, provided that $B_1^{\Ii'}, \dots,
B_k^{\Ii'}$ is a partition of $\Delta^{\Ii'}$. We say that $d \in B_i^{\Ii'}$
has color $B_i$.  A coloring $\Ii'$ is \emph{$n$-proper} if
for each $d \in \Delta^{\Ii'}$ all elements of $N_n^{\Ii'}(d)$ have
different colors.
\end{definition}


\begin{lemma} \label{lem:coloring}
If $\Ii\setminus \Delta^\Mm$ has bounded degree, then for all $n\geq
0$ there exists an $n$-proper coloring of $\Ii$ with finitely many
colors.
\end{lemma}

\begin{definition}\label{def:bounded}
An interpretation $\Ii$ is \emph{$\ell$-bounded} if for
each $r\in\troles$, each simple $r$-path has length at most $\ell$.
\end{definition}

The following statement combines several steps established by
\citeauthor{GogaczIM18}~[\citeyear{GogaczIM18}].

\begin{theorem} \label{thm:blocking}
Let $\psi$ be a UCQ. Let $\ell, n \in \mathbb{N}$ and let 
$\Ii$, $\Ii'$, and $\Jj$ be interpretations such that 
\begin{enumerate}
  \item $\Ii\setminus \Delta^\Mm$ is $\ell$-bounded and has bounded degree;
\item $\Ii'$ is an $n$-proper coloring of $\Ii$ with finitely many colors;
\item $\Delta^\Jj \subseteq \Delta^\Ii$, $A^\Jj = A^\Ii\cap \Delta^\Jj$ for 
all $A\in\mn{N_C}$, and for all $(d,e)\in r^\Jj \setminus r^\Ii$ with
$r\in\mn{N_R}$, there exists $e'$ such that $(d,e')\in r^{\Ii'}$ and
$N^{\Ii'}_n(e)\simeq N^{\Ii'}_n(e')$;
\item $\Jj\setminus \Delta^\Mm$ is $\ell$-bounded.
\end{enumerate}
If $n$ is large enough with respect to $\ell$, $|\psi|$, and
$|\Delta^\Mm|$, then $\Ii^* \notmodels \varphi$ implies $\Jj^*
\notmodels \varphi$.
\end{theorem}

Let us now apply this to our counter-witness $\Ii$ with a strongly
canonical tree decomposition $\Smf$ of bounded width and outdegree.

The first condition in Theorem~\ref{thm:blocking} is ensured by 
the properties of $\Smf$. Indeed, for each element $d$ outside of
$\Delta^\Mm$, the degree is bounded by the sum of the sizes of the
bags containing $d$. But by strong canonicity each such $d$ occurs
only in the bag where it is fresh, and a subset of its
children. Because $\Smf$ has bounded degree and width, it follows that
the degree of each such $d$ is also bounded. For $\ell$-boundedness it
suffices to notice that after removing $\Delta^\Mm$, all bags
corresponding to a transitive role name $r$ are
disjoint. Consequently, for $\ell$ we can take any number bounding the
size of bags in $\Smf$.

By Lemma~\ref{lem:coloring}, for any $n$ there exists an $n$-proper
coloring $\Ii'$ of $\Ii$ with finitely many colors. 
We construct $\Jj$ from $\Ii$ like before, but for $v$ and $u$ we
additionally require that the isomorphism $h$ witnessing $\Imf(v)
\simeq \Imf(u)$ preserves $n$ neighborhoods in $\Ii'$: for each $d\in
\Delta_v$, $N_n^{\Ii'}(d) \simeq N_n^{\Ii'}(h(d))$. One can find such
$v$ and $u$ on each path because $\Smf$ has bounded width and the size
of neighborhoods of radius $n$ is bounded as well; the latter holds
because $\Ii\setminus \Delta^\Mm$ has bounded degree.  This additional
requirement ensures the third condition in Theorem~\ref{thm:blocking}.

We claim that the length of the longest simple $r$-path avoiding
$\Delta^\Mm$ for any transitive role name $r$ can only increase by
one.  Indeed, let us examine what happens when we redirect $r$-edges
from $d_v$ with $\bago(v)=r$. As $d_v\notin \Delta^\Mm$, by the strong
canonicity of $\Smf$ we know that $v$ is the only bag storing
$r$-edges that contains $d_v$. Consequently, all $r$-edges that enter
$d_v$ in $\Ii$ originate in $\Delta_v$. Previous steps of the
procedure might have redirected some $r$-edges to $d_v$, but all these
edges originate in elements from the subtree of $\Smf$ rooted at $v$,
and these elements do not belong to $\Delta^\Mm$. Hence, when $v$ is
processed, all these edges disappear, because their origins are
removed. Because we only redirect $r$-edges to bags that have not been
replaced by a single element before, no $r$-edge will ever be
redirected to $d_v$. Thus, we have a global property that each
$r$-edge in $\Jj$ that is a result of a redirection originates in an
element that has no incoming $r$-edges. This completes the proof of
the claim. Let us take for $\ell$ in Theorem~\ref{thm:blocking} the
maximal size of a bag in $\Smf$ plus 1. 

Finally, rewriting our PEQ $\varphi$ as a 
UCQ $\psi$, we can conclude that if $n$ is sufficiently large,
$\Jj^*\notmodels \varphi$.


%% file: finsat-siq-appendix.tex
\section{Missing proofs for Section~\ref{sec:finsat-siq}}

We formulate our result slightly stronger in terms of concept
satisfiability, which is the problem of deciding, given $C,\Tmc$,
whether there is a model \Imc of $\Tmc$ with $C^\Imc\neq\emptyset$.

\begin{restatable}{lemma}{lemsiqhard}\label{lem:siqhard}
  Both finite and unrestricted concept satisfiability relative to
  \SIQfwd TBoxes over one transitive role are \twoexp-hard.
\end{restatable}
\begin{proof} 
We reduce the word problem for exponentially space bounded alternating
Turing machines (ATMs).  We actually use a slightly unusual ATM model
which is easily seen to be equivalent to the standard model.  

An \emph{alternating Turing machine (ATM)} is a tuple
$M=(Q,\Theta,\Gamma,q_0,\Delta)$ where $Q=Q_{\exists}\uplus
Q_{\forall}$ is the set of states that consists of \emph{existential
states} in~$Q_{\exists}$ and \emph{universal states} in~$Q_{\forall}$.
Further, $\Theta$ is the input alphabet and $\Gamma$ is the tape
alphabet that contains a \emph{blank symbol} $\Box \notin \Theta$,
$q_0\in Q_{\exists}$ is the \emph{starting state}, and the
\emph{transition relation} $\Delta$ is of the form
$\Delta\subseteq Q\times \Gamma\times Q\times \Gamma \times \{L,R\}.$
The set $\Delta(q,\sigma):=\{(q',\sigma',M)\mid
(q,\sigma,q',\sigma',M)\in\Delta\}$ must contain exactly two elements
for every $q \in Q$ and $\sigma \in \Gamma$.  Moreover, the state $q'$
must be from $Q_\forall$ if $q \in Q_\exists$ and from $Q_\exists$
otherwise, that is, existential and universal states alternate.  Thus,
for every configuration, there are precisely two successor
configurations; we refer to them using the \emph{first} and
\emph{second} successor configurations by fixing an (arbitrary) order
on $\Delta(q,\sigma)$.  Note that there is no accepting state. The ATM
accepts if it runs forever and rejects otherwise. Starting from the
standard ATM model, this can be achieved by assuming that
exponentially space bounded ATMs terminate on any input and then
modifying them to enter an infinite loop from the accepting state.

A \emph{configuration} of an ATM is a word $wqw'$ with
\mbox{$w,w'\in\Gamma^*$} and $q\in Q$.  We say that $wqw'$ is
\emph{existential} if~$q$ is, and likewise for \emph{universal}.
\emph{Successor configurations} are defined in the usual way.  Note
that every configuration has exactly two successor configurations. We
call these the \emph{left} and \emph{right} successor configurations.

A \emph{computation tree} of an ATM $M$ on input $w$ is an infinite
tree without leafs whose nodes are labeled with configurations of $M$
such that
\begin{itemize}

  \item the root is labeled with the initial configuration
    $q_0w$;

  \item if an inner node is labeled with an existential configuration 
    $wqw'$, then it has a single successor and this successor is labeled 
    with a successor configuration of $wqw'$;

  \item if an inner node is labeled with a universal configuration
    $wqw'$, then it has two successors and these successors are
    labeled with the two successor configurations of~$wqw'$.

\end{itemize}
An ATM $M$ \emph{accepts} an input $w$ if there is a computation
tree of $M$ on $w$. 

There is a fixed $2^n$ space bounded ATM $M$ whose word problem is
\twoexp-hard~\cite{chandraAlternation1981}. We assume that $M$ has been
modified so that it never attempts to move left on the left-most tape
cell and so that we are only interested in non-empty inputs $w$.

Given $M$ and $w$, we construct a TBox $\Tmc_{M,w}$ using a concept
name $I$ and a single transitive role $r$ such that $M$ accepts $w$
iff $I$ is satisfiable relative to $\Tmc_{M,w}$. 

We refrain from repeating the concept inclusions given in the main
part, and rather address Points~\textit{(i)} and~\textit{(iii)}, and
provide the missing details for Point~\textit{(ii)}. 

\smallskip For Point~\textit{(i)}, we first ensure that there is
exactly one leaf labeled with a symbol of shape $(q,a)$, by using a
fresh concept name $\mn{Head}$ and including
the following CIs: 
\begin{align*}
  L_0 & \sqsubseteq \qnrleq{(2^n - 2)}{r}{(\bigsqcup_{i=1}^nL_i)} \\
  L_0 & \sqsubseteq \qnreq{1}{r}{(L_n\sqcap \mn{Head})} \\
  %
  L_n \sqcap \mn{Head} & \sqsubseteq \bigsqcup_{(q,a)\in \Sigma}
  C^h_{(q,a)} \\
  L_n \sqcap \neg\mn{Head} & \sqsubseteq \bigsqcap_{(q,a)\in
  \Sigma}\neg C^h_{(q,a)} 
\end{align*}
Intuitively, the first CI restricts the number of nodes in the tree,
so that the enforced tree is \emph{exactly} the expected full binary
tree. The second CI states that there is \emph{exactly} one leaf
satisfying $\mn{Head}$. The remaining CIs enforce that the leaf
satisfying \mn{Head} indeed satisfies some $C_{(q,a)}^h$, and that the others
do not satisfy such a concept. 
%

It remains to synchronize neighboring leaves according their labeling
with symbols $C^x_\sigma$. For doing so, we need to
make precise how numbers are associated to leaves. A leaf node $d$ has
value $i\in\{0,\ldots,2^{n}-1\}$ precisely if, for all $j$ with $1\leq
j\leq n$, the $j$-th bit in the binary encoding of $i$ is $1$ iff $d$
satisfies $X_i$. As a convention, we assume that $X_0$ is responsible
for the least significant bit. It will be convenient to use the
abbreviations $X_i^*$ and $X_i^+$, $0\leq i\leq n-1$, for 
\begin{align*} 
  \neg X_i \sqcap \bigsqcap_{0\leq k\leq i-1}X_k \quad\text{and}\quad
  X_i \sqcap \bigsqcap_{0\leq k\leq i-1} \neg X_k,
\end{align*}
respectively.  To synchronize leaves with consecutive numbers, we
enforce additional $r$-successors as follows. We introduce another set
of concept names: $H$ and $D_{\sigma},D_{\sigma}'$, for every
$\sigma\in \Sigma$. The idea is to introduce for a leaf with value $i$
and labeling $C_{\sigma_l}^l,C_{\sigma_h}^h,C_{\sigma_r}^r$ two
$r$-successors: one with value $i+1$ satisfying $D_{\sigma_h}$ and
$D'_{\sigma_r}$, and one with value $i$ satisfying $D_{\sigma_l}$ and
$D'_{\sigma_h}$. We additionally require that all these successors
satisfy $H$.  This is realized using the following CIs, for every
$\sigma,\sigma'\in \Sigma$, every $i$ with $0\leq i<n$, and every
$j>i$:
\begin{align*}
  L_n \sqcap X_i^* \sqcap C^h_\sigma\sqcap C_{\sigma'}^r & \sqsubseteq
  \exists r.(H \sqcap X_i^+\sqcap D_{\sigma}\sqcap D'_{\sigma'}) \\
  L_n \sqcap X_i^* \sqcap C^l_\sigma\sqcap C_{\sigma'}^h & \sqsubseteq
  \exists r.(H\sqcap X_i^*\sqcap D_{\sigma}\sqcap D'_{\sigma'}) \\
  L_n\sqcap X_i^* \sqcap X_j & \sqsubseteq \forall r.(H\rightarrow
  X_j)\\
  L_n\sqcap X_i^* \sqcap \neg X_j & \sqsubseteq \forall
  r.(H\rightarrow \neg X_j) \\
  D_\sigma & \sqsubseteq \neg D_{\sigma'} \quad \text{if
  $\sigma\neq\sigma'$} \\
  D'_\sigma & \sqsubseteq \neg D'_{\sigma'} \quad \text{if
  $\sigma\neq\sigma'$}
\end{align*}
Now, to enforce the synchronization, we add the concept inclusion
$$L_0  \sqsubseteq \qnrleq{(2^n-2)}{r}{H},$$
forcing some of the newly created successors to identify, which is
only possible if they have the same number and the same labeling with
the $D_\sigma,D_{\sigma}'$.  Using the provided intuitions it is not
difficult to verify that: 

\smallskip\noindent\textit{Claim~1.} If $\Imc$ is a model of the TBox
constructed so far, $d\in L_0^\Imc$, and $d_0,\ldots,d_{2^n-1}$ are
leaf elements reachable from $d$, then $d_0,\ldots,d_{2^n-1}$
represent a valid configuration of $M$.

\medskip For Point~\textit{(ii)}, that is, synchronization of
successor configurations, we create at the leaves of every enforced
tree successors which contain the current configuration and the
successor configuration(s) as described in the body of the paper.
Using at-most restrictions enforced in elements satisfying
$B_\exists,B_{\forall,1},B_{\forall,2}$, we force corresponding
successors to join and thus to synchronize. We use additional
concept names $E,S_\exists,S_\forall^1,S_\forall^2$.
Notice that we need different versions of the concept $S$ mentioned in
the body of the paper, to synchronize between different parts of the
computation, and moreover that $E$ is used to make a non-deterministic
choice for existential configurations. We include the following
concept inclusions, for all $i\in\{1,2\}$, and all
$\sigma_l,\sigma_h,\sigma_r\in \Sigma$ which can occur in neighboring
cells in a valid configuration of $M$: 
\begin{align*}
  %
  %
  L_0 & \sqsubseteq\forall r.E \sqcup \forall r.\neg E \\
  L_n \sqcap \exists r^-.A_{\exists}^i\sqcap \langle
  \sigma_l,\sigma_h,\sigma_r\rangle \sqcap \neg E & \sqsubseteq
  \exists r.(S_\exists\sqcap
  D_{\langle\sigma_l,\sigma_h,\sigma_r\rangle}^1) \\
  L_n \sqcap \exists r^-.A_{\exists}^i\sqcap \langle
  \sigma_l,\sigma_h,\sigma_r\rangle \sqcap E & \sqsubseteq \exists
  r.(S_\exists\sqcap D_{\langle\sigma_l,\sigma_h,\sigma_r\rangle}^2)
  \\
  L_n \sqcap \exists r^-.A_\forall\sqcap \langle
  \sigma_l,\sigma_h,\sigma_r\rangle& \sqsubseteq \exists
  r.(S_\forall^1\sqcap
  D^1_{\langle\sigma_l,\sigma_h,\sigma_r\rangle}) \\
  L_n \sqcap \exists r^-.A_\forall\sqcap  \langle
  \sigma_l,\sigma_h,\sigma_r\rangle & \sqsubseteq  \exists
  r.(S_\forall^2\sqcap
  D^2_{\langle\sigma_l,\sigma_h,\sigma_r\rangle}) \\
  L_n \sqcap \exists r^-.A^i_{\exists}\sqcap C^h_{\sigma_h} &
  \sqsubseteq \exists r.(S_\forall^i\sqcap D_{\sigma_h}) \\ 
  L_n \sqcap \exists r^-.A_{\forall}\sqcap C^h_{\sigma_h} &
  \sqsubseteq \exists r.(S_\exists\sqcap D_{\sigma_h}) \\
  S_\forall^i & \sqsubseteq \neg S_\forall^{3-i} \\
  S_\forall^i & \sqsubseteq \neg S_\exists
\end{align*}
where $\langle \sigma_l,\sigma_h,\sigma_r\rangle$ is an abbreviation
for $C^l_{\sigma_l}\sqcap C^h_{\sigma_h}\sqcap C^r_{\sigma_r}$, and
$D_{\langle \sigma_l,\sigma_h,\sigma_r\rangle}^j$, $j\in \{1,2\}$ is
$D_\sigma$ if, when in a configuration with neighboring cells
$\sigma_l,\sigma_h,\sigma_r$, the next label of the current cell in
the $j$-th successor configuration is $\sigma$ (recall that every
configuration has exactly two successor configurations).  

Similar to the what was done before, we propagate the values
associated to the leaves to the newly created successors. This is done
using the following CIs, for every $i$ with $0\leq i\leq n-1$ and $S\in
\{S_\exists, S_\forall^1, S_{\forall}^2\}$: 
\begin{align*}
  L_n \sqcap X_j & \sqsubseteq \forall r.(S \rightarrow X_j ) \\
  L_n \sqcap \neg X_j & \sqsubseteq \forall r.(S \rightarrow \neg X_j )
\end{align*}

It remains to give the promised at-most restrictions, for all $i\in
\{1,2\}$:
\begin{align*}
  B_\exists & \sqsubseteq \qnrleq{2^n}{r}{S_\exists}, &&
  B_{\forall}^i \sqsubseteq \qnrleq{2^n}{r}{S_{\forall}^i}.
\end{align*}

Again, based on the provided intuitions it is not difficult to verify
that neighboring configurations are indeed successor configurations
according to $M$'s transition relation. More precisely, we have:

\smallskip\noindent\textit{Claim~2.} If \Imc is a model of the
TBox constructed so far, then for all $d,e,e'\in \Delta^\Imc$ such
that $(d,e),(d,e')\in r^\Imc$ we have:
\begin{enumerate}

  \item if $d\in B_\exists^\Imc$, $e\in (A_\exists^i)^\Imc$ for some
    $i\in \{1,2\}$, and $e'\in A_{\forall}^\Imc$, then the
    configuration below $e'$ is the first successor
    configuration of the configuration below $e$ if $e'\in (\forall
    r.E)^\Imc$; and the second successor
    configuration of the configuration below $e$ if $e'\in (\forall
    r.\neg E)^\Imc$;

  \item if, for some $i\in \{1,2\}$, $d\in (B_\forall^i)^\Imc$, $e\in A_\forall^\Imc$, and
    $e'\in (A_{\exists}^i)^\Imc$, then the configuration below $e'$ is
    the $i$-th successor configuration of the configuration below $e$.

\end{enumerate}

\medskip For Point~\textit{(iii)}, that is, the enforcement of the
initial configuration, let $w=a_0\cdots a_{n-1}$ be the input word.
It is routine to give (polynomially sized) concepts
$(X=k)$, for $0\leq k\leq n-1$, describing all leaves with value $k$,
and $(X\geq n)$ describing all nodes with value $\geq n$. We include
the following concept inclusions, for all $i$ with $0<i<n$:
\begin{align*}
  L_n \sqcap \exists r^-.I\sqcap (X=0) & \sqsubseteq C_{(q_0,a_0)}^h \\
  L_n \sqcap \exists r^-.I\sqcap (X=i) & \sqsubseteq C_{a_i}^h \\
  L_n \sqcap \exists r^-.I\sqcap (X\geq n) & \sqsubseteq C_{\Box}^h.
\end{align*}
Recall that $\Box$ is the blank symbol in the tape alphabet. It is not
difficult to see that the configuration encoded in the tree starting
from $I$ is the initial configuration.

\smallskip
This finishes the construction of the TBox. Correctness of the
reduction is established in the following. 

\smallskip\noindent\textit{Claim~3.} $M$ accepts $w$ iff $I$ is
satisfiable relative to $\Tmc_{M,w}$.

\smallskip\noindent\textit{Proof of Claim~3.} $(\Rightarrow)$ If $M$ accepts $w$, there is an infinite
alternating computation of $M$ on input $w$. We inductively convert
the computation into an interpretation $\Imc$ in the expected way:
\begin{itemize}

  \item Start with a tree whose leaves are labeled with the
    initial configuration, and whose root is labeled with
    $I,A_\exists^1,L_0$. 

  \item Choose some node $d$ in the interpretation constructed so far
    satisfying either $A_\forall$ or $A_{\exists}^i$ and let
    $\alpha$ be the configuration the leaves of the tree below $d$.

    \begin{itemize}

      \item if $d$ satisfies $A_\forall$, we inductively know that
	$\alpha$ is a universal configuration and has two successor
	configurations $\alpha_1,\alpha_2$ in the accepting
	computation. We add new elements $e_1,e_2,f_1,f_2$ to \Imc
	such that, for $i\in\{1,2\}$, $e_i\in (B_{\forall}^i)^\Imc$,
	$f_i\in (A_{\exists}^i\sqcap L_0)^\Imc$, and $(e_i,d)\in
	r^\Imc$, and $(e_i,f_i)\in r^\Imc$. Moreover, add below
	$e_i$ a tree with configuration $\alpha_i$ in the leaves. 

      \item if $d$ satisfies $A_\exists^i$, we inductively know that
	$\alpha$ is an existential configuration and has one successor
	configurations $\alpha'$ in the accepting computation. If
	$\alpha'$ is the second successor of $\alpha$ (according to
	$\Delta$), we add all elements in the tree below $d$ to
	$E^\Imc$.  Additionally, add new elements $e,f$ to \Imc
	such that $e\in B_{\exists}^\Imc$,
	$f\in (A_{\forall}\sqcap L_0)^\Imc$, and $(e,d)\in
	r^\Imc$, and $(e,f)\in r^\Imc$. Moreover, add below
	$e$ a tree with configuration $\alpha'$ in the leaves. 

    \end{itemize}

\end{itemize}
It is routine to verify that the interpretation $\Imc$ obtained in the
limit is a model of $I$ and $\Tmc_{M,w}$.

\smallskip $(\Leftarrow)$ Let $\Imc$ be a model of $\Tmc_{M,w}$ and
$d_0\in I^\Imc$. By construction of $\Tmc_{M,w}$ (in particular, the
concept inclusions starting with $I$ from the body of the paper),
there is an infinite tree $(T,\tau)$ labeled with elements from $\Delta^\Imc$
and having the following properties: 
\begin{itemize}

  \item the root node is labeled with $d_0$, that is, $\tau(\varepsilon)=d$;

  \item $T$ has outdegree one in odd levels (assuming $d$ is in the
    first level) and outdegree two in even levels;

  \item nodes $n$ in odd levels of $T$ satisfy $\tau(n)\in
    (A_\exists^i\sqcap L_0)^\Imc$, for some $i\in\{1,2\}$, and have
    a single successor $n'$ such that there is some $e\in \Delta^\Imc$
    with $e\in B_\exists^\Imc$ and $(e,\tau(n)),(\tau(n'),e)\in
    r^\Imc$;
    
  \item nodes $n$ in even levels of $T$ satisfy $\tau(n)\in
    (A_\forall \sqcap L_0)^\Imc$ and have
    two successors $n_1,n_2$ such that there are $e_1,e_2\in \Delta^\Imc$
    with $e_i\in (B_\forall^i)^\Imc$, $(e_i,\tau(n)),(\tau(n'),e_i)\in
    r^\Imc$, and $\tau(n_i)\in (A_\exists^i)^\Imc$;

\end{itemize}
Note that $T$ has the structure of an infinite alternating computation
of $M$. It remains to associate a configuration to every node in the
tree. By construction of $T$, we have $\tau(n)\in L_0^\Imc$, for all
nodes $n\in T$. By Claim~1, there is a tree labeled with a
configuration below each $\tau(n)$.

By Point~\textit{(iii)}, we know that the configuration in the below
$d_0$ is the initial configuration $\alpha_0$ of $M$ on input $w$. We
can now use Claim~2 and the construction of $T$ to inductively
construct an infinite, alternating computation of $M$ on input
$w$. Thus $M$ accepts $w$.

\medskip This finishes the proof of the claim, and in fact of
\twoexp-hardness in the unrestricted case. For finite
reasoning, we have to show that $I$ is satisfiable relative to
$\Tmc_{M,w}$ iff it is finitely satisfiable relative to
$\Tmc_{M,w}$. Since the ``if''-direction is trivial, we focus on the
``only if''-direction.

It suffices to show that there is a finite model \Imc of $I$ and
$\Tmc_{M,w}$ in case $M$ accepts $w$. We can vary the construction
in the $(\Rightarrow)$-direction of the proof of Claim~3 as follows.
Instead of adding new elements in every step of the inductive
construction, we can reuse old elements in case the new element is
associated to the same configuration as the old one. Since $M$ is
space bounded, this will happen on every possible path, and hence we
end up with a finite model. 
\end{proof}

%
%
%

%% file: arxiv.bbl
\begin{thebibliography}{}

\bibitem[\protect\citeauthoryear{Amarilli and Benedikt}{2015}]{AmarilliB15}
Antoine Amarilli and Michael Benedikt.
\newblock Finite open-world query answering with number restrictions.
\newblock In {\em {LICS}-15}, pages 305--316, 2015.

\bibitem[\protect\citeauthoryear{Baader \bgroup \em et al.\egroup
  }{2017}]{DLBook}
Franz Baader, Ian Horrocks, Carsten Lutz, and Ulrike Sattler.
\newblock {\em An Introduction to Description Logic}.
\newblock Cambridge University Press, 2017.

\bibitem[\protect\citeauthoryear{Bednarczyk \bgroup \em et al.\egroup
  }{2019}]{BKieronski18}
Bartosz Bednarczyk, Emanuel Kieronski, and Piotr Witkowski.
\newblock On the complexity of graded modal logics with converse.
\newblock In {\em {JELIA}}, volume 11468 of {\em LNCS}, pages 642--658, 2019.

\bibitem[\protect\citeauthoryear{Bienvenu and Ortiz}{2015}]{BienvenuO15}
Meghyn Bienvenu and Magdalena Ortiz.
\newblock Ontology-mediated query answering with data-tractable description
  logics.
\newblock In {\em Proc.\ of {RW}-15}, pages 218--307, 2015.

\bibitem[\protect\citeauthoryear{Calvanese \bgroup \em et al.\egroup
  }{2014}]{CalvaneseEO14}
Diego Calvanese, Thomas Eiter, and Magdalena Ortiz.
\newblock Answering regular path queries in expressive description logics via
  alternating tree-automata.
\newblock {\em Inf. Comput.}, 237:12--55, 2014.

\bibitem[\protect\citeauthoryear{Chandra \bgroup \em et al.\egroup
  }{1981}]{chandraAlternation1981}
Ashok~K. Chandra, Dexter~C. Kozen, and Larry~J. Stockmeyer.
\newblock Alternation.
\newblock {\em Journal of the ACM}, 28:114--133, 1981.

\bibitem[\protect\citeauthoryear{Danielski and Kieronski}{2018}]{Kieronski18}
Daniel Danielski and Emanuel Kieronski.
\newblock Finite satisfiability of unary negation fragment with transitivity.
\newblock {\em arXiv:1809.03245}, 2018.

\bibitem[\protect\citeauthoryear{Glimm \bgroup \em et al.\egroup
  }{2008a}]{GlimmHS08}
Birte Glimm, Ian Horrocks, and Ulrike Sattler.
\newblock Unions of conjunctive queries in {SHOQ}.
\newblock In {\em Proc.\ of {KR}-08}, pages 252--262, 2008.

\bibitem[\protect\citeauthoryear{Glimm \bgroup \em et al.\egroup
  }{2008b}]{GlimmLHS08}
Birte Glimm, Carsten Lutz, Ian Horrocks, and Ulrike Sattler.
\newblock Conjunctive query answering for the description logic {SHIQ}.
\newblock {\em J. Artif. Intell. Res. {(JAIR)}}, 31:157--204, 2008.

\bibitem[\protect\citeauthoryear{Gogacz \bgroup \em et al.\egroup
  }{2018}]{GogaczIM18}
Tomasz Gogacz, Yazmin Ib{\'{a}}{\~{n}}ez{-}Garc{\'{\i}}a, and Filip Murlak.
\newblock Finite query answering in expressive description logics with
  transitive roles.
\newblock In {\em Proc.\ of {KR}-18}, 2018.

\bibitem[\protect\citeauthoryear{Gr{\"{a}}del and
  Walukiewicz}{1999}]{GradelW99}
Erich Gr{\"{a}}del and Igor Walukiewicz.
\newblock Guarded fixed point logic.
\newblock In {\em Proc.\ of {LICS}-99}, pages 45--54, 1999.

\bibitem[\protect\citeauthoryear{Gr{\"{a}}del \bgroup \em et al.\egroup
  }{2002}]{2001automata}
Erich Gr{\"{a}}del, Wolfgang Thomas, and Thomas Wilke, editors.
\newblock {\em Automata, Logics, and Infinite Games: {A} Guide to Current
  Research}, volume 2500 of {\em {LNCS}}, 2002.

\bibitem[\protect\citeauthoryear{Guti{\'e}rrez-Basulto \bgroup \em et
  al.\egroup }{2017}]{GuIbJu-AAAI17}
V{\'i}ctor Guti{\'e}rrez-Basulto, Yazm{\'i}n Ib{\'a}{\~n}ez-Garc{\'i}a, and
  Jean~Christoph Jung.
\newblock Number restrictions on transitive roles in description logics with
  nominals.
\newblock In {\em Proc.\ of {AAAI-17}}, 2017.

\bibitem[\protect\citeauthoryear{Guti{\'e}rrez-Basulto \bgroup \em et
  al.\egroup }{2018}]{GuIbJu-AAAI18}
V{\'i}ctor Guti{\'e}rrez-Basulto, Yazm{\'i}n Ib{\'a}{\~n}ez-Garc{\'i}a, and
  Jean~Christoph Jung.
\newblock Answering regular path queries over sq ontologies.
\newblock In {\em Proc.\ of {AAAI-18}}. AAAI Press, 2018.

\bibitem[\protect\citeauthoryear{Ib{\'{a}}{\~{n}}ez{-}Garc{\'{\i}}a \bgroup \em
  et al.\egroup }{2014}]{GarciaLS14}
Yazmin Ib{\'{a}}{\~{n}}ez{-}Garc{\'{\i}}a, Carsten Lutz, and Thomas Schneider.
\newblock Finite model reasoning in horn description logics.
\newblock In {\em Proc.\ of {KR}-14}, 2014.

\bibitem[\protect\citeauthoryear{Kaminski and Smolka}{2010}]{KaminskiS10}
Mark Kaminski and Gert Smolka.
\newblock Terminating tableaux for $\mathcal{SOQ}$ with number restrictions on
  transitive roles.
\newblock In {\em Proc.\ of the 6th {IFIP} {TC}}, pages 213--228, 2010.

\bibitem[\protect\citeauthoryear{Kazakov and
  Pratt{-}Hartmann}{2009}]{KazakovP09}
Yevgeny Kazakov and Ian Pratt{-}Hartmann.
\newblock A note on the complexity of the satisfiability problem for graded
  modal logics.
\newblock In {\em Proc.\ of {LICS}-09}, pages 407--416, 2009.

\bibitem[\protect\citeauthoryear{Kazakov \bgroup \em et al.\egroup
  }{2007}]{DBLP:conf/lpar/KazakovSZ07}
Yevgeny Kazakov, Ulrike Sattler, and Evgeny Zolin.
\newblock How many legs do {I} have? {N}on-simple roles in number restrictions
  revisited.
\newblock In {\em Proc.\ of {LPAR}-07}, pages 303--317, 2007.

\bibitem[\protect\citeauthoryear{Kontchakov and
  Zakharyaschev}{2014}]{KontchakovZ14}
Roman Kontchakov and Michael Zakharyaschev.
\newblock An introduction to description logics and query rewriting.
\newblock In {\em Proc.\ of {RW}-14}, pages 195--244, 2014.

\bibitem[\protect\citeauthoryear{Lutz}{2008}]{DBLP:conf/cade/Lutz08}
Carsten Lutz.
\newblock The complexity of conjunctive query answering in expressive
  description logics.
\newblock In {\em Proceedings of {IJCAR} 2008}, pages 179--193, 2008.

\bibitem[\protect\citeauthoryear{Ngo \bgroup \em et al.\egroup
  }{2016}]{DBLP:conf/kr/NgoOS16}
Nhung Ngo, Magdalena Ortiz, and Mantas Simkus.
\newblock Closed predicates in description logics: Results on combined
  complexity.
\newblock In {\em Proc.\ of {KR}-16}, pages 237--246, 2016.

\bibitem[\protect\citeauthoryear{Pratt{-}Hartmann}{2009}]{Pratt-Hartmann09}
Ian Pratt{-}Hartmann.
\newblock Data-complexity of the two-variable fragment with counting
  quantifiers.
\newblock {\em Inf. Comput.}, 207(8):867--888, 2009.

\bibitem[\protect\citeauthoryear{Rosati}{2008}]{Rosati08}
Riccardo Rosati.
\newblock Finite model reasoning in dl-lite.
\newblock In {\em Proc.\ of {ESWC}-19}, pages 215--229, 2008.

\bibitem[\protect\citeauthoryear{Rudolph}{2016}]{Rudolph16}
Sebastian Rudolph.
\newblock Undecidability results for database-inspired reasoning problems in
  very expressive description logics.
\newblock In {\em Proc.\ of {KR}-16}, pages 247--257, 2016.

\end{thebibliography}
